\documentclass[a4paper,11pt]{article}
\pdfoutput=1
\usepackage[margin=1truein]{geometry}

\usepackage{lmodern}
\usepackage[T1]{fontenc}
\usepackage{textcomp}

\usepackage{graphicx}
\usepackage{amsmath, amssymb}
\usepackage{amsthm}
\usepackage{thmtools}
\usepackage{mathrsfs}
\usepackage{authblk}
\usepackage{bm}
\usepackage{cases}
\usepackage{color}

\usepackage[
    bookmarksnumbered,
    colorlinks,
    hypertexnames=false,
    pdfdisplaydoctitle,
    unicode
]{hyperref}
\usepackage{algorithm}
\usepackage{algorithmicx}
\usepackage{algpseudocode}

\usepackage{comment}
\usepackage{url}
\usepackage{multicol}
\usepackage{siunitx}
\usepackage{mathtools}
\usepackage{nicematrix}
\usepackage{physics}
\DeclarePairedDelimiterX{\Set}[2]{\{}{\}}{\;#1\mathclose{}\nonscript\;\delimsize|\nonscript\;\mathopen{}#2\;}
\DeclarePairedDelimiter{\prn}{(}{)}
\DeclareMathOperator{\diag}{diag}

\DeclareMathOperator{\trank}{term-rank}
\DeclareMathOperator{\e}{e}

\usepackage[utf8]{inputenc}
\newtheorem{theorem}{Theorem}[section]
\newtheorem{lemma}[theorem]{Lemma}
\newtheorem{proposition}[theorem]{Proposition}

\theoremstyle{definition}
\newtheorem{definition}[theorem]{Definition}
\newtheorem{example}[theorem]{Example}
\newtheorem{remark}[theorem]{Remark}

\usepackage{latexsym}
\usepackage{mathtools}
\usepackage{nicematrix}
\makeatletter
\renewcommand{\ALG@name}{Algorithm}
\makeatother

\usepackage{booktabs}
\usepackage{enumitem}

\usepackage{setspace}
\usepackage{todonotes}

\usepackage{url}

\title{Structural Preprocessing Method for Nonlinear Differential-Algebraic Equations Using Linear Symbolic Matrices}
\author[1]{Taihei Oki}
\author[1]{Yujin Song}
\affil[1]{\small{Department of Mathematical Informatics, Graduate School of Information Science and Technology, University
of Tokyo, Tokyo 113-8656, Japan.}}
\date{March 2024}

\begin{document}

\maketitle
\begin{abstract}
  Differential-algebraic equations (DAEs) have been used in modeling various dynamical systems in science and engineering.
  There are several preprocessing methods that are needed before performing numerical simulations for DAEs, such as consistent initialization and index reduction.  Preprocessing methods that use structural information on DAEs run fast and are widely used.
  Unfortunately, structural preprocessing methods may fail when the system Jacobian, which is a functional matrix, derived from the DAE is singular.  

  To transform a DAE with a singular system Jacobian into a nonsingular system, several regularization methods have been proposed, which are based on combinatorial relaxation.
  Most of all existing regularization methods rely on symbolic computation to eliminate the system Jacobian for finding a certificate of singularity, resulting in much computational time.
  Iwata--Oki--Takamatsu~(2019) proposed a method (IOT-method) to find a certificate without symbolic computations.
  The IOT method approximates the system Jacobian by a simpler symbolic matrix, called a layered mixed matrix, which admits a fast combinatorial algorithm for singularity testing.
  Although the IOT method runs fast, it often overlooks the singularity of the system Jacobian since the approximation largely discards algebraic relationships among entries in the original system Jacobian.
     
  In this study, we propose a new regularization method extending the idea of the IOT method.
  Instead of layered mixed matrices, our method approximates the system Jacobian by more expressive symbolic matrices, called rank-1 coefficient mixed (1CM) matrices.  This makes our method more widely applicable than the existing method.
  We give a fast combinatorial algorithm for finding a singularity certificate of 1CM-matrices, making our regularization method free from symbolic elimination. 
  Our method is also advantageous in that it globally preserves the solution set to the DAE.
  Through numerical experiments, we confirmed that our method runs fast for large-scale DAEs from real instances.
\end{abstract}

\section{Introduction}\label{sec:introduction}
A $k$th-order \emph{differential-algebraic equation}~(DAE) for a variable $x:\mathbb{R}\to \mathbb{R}^n$ is a system of differential equations written as
\begin{equation}\label{dae}
	f(x(t),\dot{x}(t),\dotsc,x^{(k)}(t),t)=0,
\end{equation}
where $f:\mathbb{R}^{n(k+1)+1}\to \mathbb{R}^n$ is a smooth function.  DAEs have been used in modeling various dynamical systems, such as electrical circuits, chemical reactions, and mechanical systems.  Solving DAEs numerically is an essential means for systems simulations~\cite{brenantext,hairertext}.

Numerically solving DAEs present several challenges compared to solving ordinary differential equations (ODEs), which are systems of equations in the form of $\dot{x}(t) = \varphi(x(t),t)$ for some function $\varphi: \mathbb{R}^{n+1} \to \mathbb{R}$.
One challenge for solving DAEs is to determine a \emph{consistent initial value}, which is a point $P=(x^*,\dot{x}^*,\dotsc,{x^*}^{(k)},t^*) \in \mathbb{R}^{n(k+1)+1}$ through which a unique solution $x(t)$ to~\eqref{dae} passes. This task is non-trivial because DAEs can have ``hidden'' algebraic constraints that appear when the presented equations are differentiated.
Even if a consistent initial value is given, we also encounter another difficulty in numerically solving \emph{high-index} DAEs.
The \emph{differentiation index}, or simply the \emph{index}, of a DAE is an indicator of the discrepancy between the DAE and ODEs.
More precisely, the index of a first-order DAE~\eqref{dae} at a point $P$ is the smallest nonnegative integer $\nu$ such that the equations obtained by combining the zeroth, first, $\ldots,\nu$th-order derivatives of each equation in~\eqref{dae} can be solved with respect to $\dot{x}(t)$ in the form of $\dot{x}(t)=\varphi(x(t),t)$ in a neighborhood of $P$.
Specifically, ODEs are DAEs with index zero.
It is empirically known to be hard to construct a general-purpose numerical scheme with reasonable accuracy for DAEs with index greater than one, and thus reducing the indices of DAEs is vital preprocessing prior to numerical integration.

Recent simulation software packages often adopt structural preprocessing methods, which use ``structural information'' on DAEs: the occurrence of each variable in each equation.
For example, Pantelides' and Pryce's consistent initialization methods~\cite{pantelides88,pryce01} determine the number of differentiations of each equation by solving the assignment problem on a bipartite graph constructed from structural information.
The Mattsson--Söderlind (MS) method~\cite{mattsson93} transforms a high-index DAE into a low-index system, based on Pryce's algorithm.
One problem common to these structural methods is that they may fail when the numerical and structural information on a DAE is inconsistent.
More precisely, the structural methods may not work correctly when the \emph{system Jacobian matrix}, a functional matrix constructed from the DAE and its structural information, is identically singular; see Section~\ref{sec_pre} for definition.

To cope with this problem, several methods have been proposed to convert DAEs with singular system Jacobians into DAEs with nonsingular system Jacobians~\cite{iwata19,oki23,tan17,wu13}; we refer to these methods as \emph{regularization methods}.  All existing regularization methods are based on the framework called \emph{combinatorial relaxation}, originally introduced by Murota~\cite{murota90} as an algorithm to compute the Puiseux-series solutions to determinantal equations.
The combinatorial relaxation iterates the following procedure until the system Jacobian $J$ gets nonsingular: (i) find a ``certificate of singularity'' of $J$ and (ii) modify the DAE using the certificate to eliminate the numerical cancellation inherent in $J$.
Here, different ``certificate of singularity'' are employed, varying with each regularization method.
For example, the LC- and ES-methods by Tan et al.~\cite{tan17} use (generally symbolic) non-zero vectors in the left and right kernels of $J$, respectively, and Oki's substitution and augmentation methods~\cite{oki23} search for a submatrix of $J$ with certain property.
To find these certificates, these methods perform the Gaussian elimination on $J$.
However, since the system Jacobians of nonlinear DAEs are functional matrices, this operation involves complicated symbolic computations for large-scale, complicated nonlinear DAEs, resulting in a prohibitive computational cost.
Furthermore, certificates of singularity may become invalid depending on points; for example, a symbolically non-zero vector in the left kernel of $J$ can be zero or ill-defined (e.g.\ a denominator in some component becomes zero) at some point.
In numerically solving a regularized DAE, approaching such a point makes the DAE ill-conditioned, necessitating the regularization to be redone with a new certificate.
This \emph{pivoting} operation further increases the total computational cost.

Addressing these issues, Iwata--Oki--Takamatsu~\cite{iwata19} proposed a regularization method (\emph{IOT method}) free from elimination of symbolic matrices.
Their method ``approximates'' a system Jacobian $J$ by a simpler symbolic matrix obtained by replacing functional (non-constant) entries with distinct symbols.
For instance, if $J$ is
\begin{equation*}\label{eq:intro-J}    J=\begin{pmatrix}1&-1&0&0\\-1&1&0&0\\0&0&\sin (x_3+x_4)&\sin (x_3+x_4)\\0&0&t&2t\end{pmatrix},
\end{equation*}
then the method constructs the following matrix:
\begin{equation*}
J_\mathrm{M}=\begin{pmatrix}1&-1&0&0\\-1&1&0&0\\0&0&\alpha_1&\alpha_2\\0&0&\alpha_3&\alpha_4\end{pmatrix},
\end{equation*}
where $\alpha_1, \dotsc, \alpha_4$ are distinct symbols.
We always have $\rank J \le \rank J_\mathrm{M}$ and the equality attains in this example.
The matrix $J_\mathrm{M}$ is a \emph{layered mixed matrix}~\cite{murota85,murotatext}, in which each symbol appears only once.
Thanks to this feature, we can find a certificate of singularity for $J_\mathrm{M}$ by solving a combinatorial optimization problem, called \emph{independent assignment}, without symbolic computations.
Specifically, the IOT method obtains a nonsingular constant matrix $U$ such that $UJ_\mathrm{M}$ as well as $UJ$ has a large all-zero submatrix, and then modifies the DAE using the certificate $U$.
Furthermore, unlike other methods, the IOT method globally retains the solutions of the DAE since $U$ is a constant matrix and then valid certificate at any point.
This \emph{global equivalence} property eliminates the need for pivoting.

The IOT method, however, replaces the different entries in the original system Jacobian $J$ with different symbols ignoring their algebraic relationships.
This may cause a failure of the method as $J_\mathrm{M}$ can be nonsingular while $J$ is singular; see Example~\ref{ex:iot-failure}.

In this paper, we propose a fast regularization method for nonlinear DAEs, extending the IOT method.
To capture more detailed algebraic relationships, our method approximates the system Jacobian by a \emph{linear symbolic matrix}: a symbolic matrix whose entries are affine expressions, where multiple occurrences of the same symbols are allowed.
That is, linear symbolic matrices have a form of $J_\mathrm{LS}=A_0+\alpha_1A_1+\cdots+\alpha_mA_m$, where $\alpha_1, \dotsc, \alpha_m$ are distinct symbols and $A_0, A_1, \dotsc, A_m$ are constant matrices.
Layered mixed matrices correspond to the case where only one entry of $A_i$ is nonzero for each $1\leq i\leq m$.
Unfortunately, for general linear symbolic matrices, finding a certificate of singularity has still been recognized as a non-trivial task in theoretical computer science~\cite{edmonds67}.
We observe that the ranks of the coefficient matrices $A_1, \dotsc, A_m$ except $A_0$ are typically low in several actual systems, including electrical circuits and control of multi-body systems.
Specifically, for some systems, the ranks of $A_1, \dotsc, A_m$ are all one.
Such linear symbolic matrices, referred to as \emph{rank-1 coefficient mixed matrices} (1CM-matrices), are more expressive than layered mixed matrices but still admit a combinatorial algorithm for singularity testing~\cite{ivanyos10}.
Based on this, our method further approximates the linear symbolic matrix $J_\mathrm{LS}$ by a 1CM-matrix $J_\mathrm{1CM}$ (if necessary) and then finds a certificate of singularity of $J_{\mathrm{1CM}}$.
We propose an $O((n+m)^3 \log(n+m))$-time algorithm for this, where we naively assume that multiplying two $n \times n$ matrices takes $O(n^3)$ time. 
Our algorithm is faster than the existing algorithm by Ivanyos et al.~\cite{ivanyos10} that runs in $O(n^5m)$ time, as we can assume $m \le n^2$ by moderate preprocessing.
Our method also possesses the global equivalence property.
Through numerical experiments, we confirmed that our method is applicable to many real instances and runs faster for large-scale DAEs.

Lastly, we would like to note that, while our method applies to a wider class of DAEs than the IOT method, it can still fail for some DAEs if the nonsingularity of $J$ and $J_\mathrm{1CM}$ are inconsistent.
For such DAEs, we need to switch to a more applicable regularization method such as the LC-, ES-, substitution, or augmentation methods~\cite{oki23,tan17}.
Nevertheless, our method offers a fast alternative that is worth trying before resorting to these slower methods.

\paragraph{Related Work.}
There are several other preprocessing methods for DAEs that are based on structural information.  The index reduction method by Unger et al.~\cite{unger95} computes the Jacobian of a DAE and reports that the index is zero if its determinant has at least one term in the expansion. This method may also fail if there is a discrepancy between the numerical and structural information.  In the $\sigma\nu$~(symbolic numeric) method by Chowdhry et al.~\cite{chowdhry04}, the singularity of the Jacobian is tested more precisely by retaining constant entries of the Jacobian.

The method by Yang et al.~\cite{yang22} regularizes polynomial DAEs without eliminating functional matrices using methodology from numerical algebraic geometry.  The method finds a specific point $P$ that satisfies the algebraic constraints of DAEs and then modifies the DAE in the same way as the augmentation method by Oki~\cite{oki23}. 

\paragraph{Organization.}
The remaining part of this paper is organized as follows.
Sections~\ref{sec_pre} and~\ref{subsec_reg} describe preliminaries on structural preprocessing methods and regularization methods for DAEs, respectively.
We present our method in Section~\ref{sec_prp} and report the overview of experimental results in Section~\ref{sec_num}.
Section~\ref{sec_conclusion} concludes this paper.

\paragraph{Notations.}
The sets of integers, non-negative integers, and reals are denoted by $\mathbb{Z}$, $\mathbb{Z}_{\ge 0}$, and $\mathbb{R}$, respectively, and we let $[n] \coloneqq \Set{i \in \mathbb{Z}}{1 \leq i \leq n}$ for $n \in \mathbb{Z}_{\ge 0}$.
The $i$th element of a vector $x$ is denoted by $x_i$ and the $(i, j)$ entry of a matrix $A$ is written as $A_{ij}$.
We denote by $\diag\{a_i\}$ the diagonal matrix whose $i$th diagonal element is $a_i$.
For a row subset $I$ and a column subset $J$, $A[I,J]$ is defined as the submatrix of $A$ obtained by extracting rows and columns in $I$ and $J$, respectively.
If $I$ is all the rows, we simply write $A[I, J]$ as $A[J]$.
The degree of a polynomial $f(s)$ is denoted by $\deg_s f(s)$.  We denote the derivative of a real function $x(t)$ by $\dot{x}(t)$ and the $l$th-order derivative by $x^{(l)}(t)$ for $l \in \mathbb{Z}_{\ge 0}$. 
For a vector-valued function $f: \mathbb{R}^m \to \mathbb{R}^n$ and $i \in [n]$, let $f_i: \mathbb{R}^m \to \mathbb{R}$ denotes the scalar-valued function with $f(x) = {(f_i(x))}_i$ for $x \in \mathbb{R}^m$.
When a function $f$ is identically zero, it is written as $f \equiv 0$.

\section{Structural Methods for DAEs}\label{sec_pre}

We describe structural methods for consistent initialization and index reduction to DAEs.
Consider a DAE $f=0$ in size $n$ given in~\eqref{dae}.
The \emph{$\sigma$-function} $\sigma_f:[n]^2 \to \mathbb{Z}_{\ge 0} \cup \{-\infty\}$ of $f$ is defined as
\begin{equation}
	\sigma_f(i,j)\coloneqq \max\Set*{l \in \mathbb{Z}_{\ge 0}}{\frac{\partial f_i }{\partial x_j^{(l)} }\not\equiv 0} \quad (i,j \in [n]),
\end{equation}
where $\max \emptyset \coloneqq -\infty$.  The \emph{structural graph} $G_f$ of $f$ is a weighted complete bipartite graph whose vertex classes are $U=\{u_1,\dotsc,u_n\},V=\{v_1,\ldots,v_n\}$ and whose edge $(u_i,v_j)$ is weighted as $\sigma_f(i,j)$ for $i,j \in [n]$.
If the DAE is a linear time-invariant DAE in the form of
\begin{equation}\label{lti}
f \coloneqq \sum _{l=0}^k A_l x^{(l)}(t)-g(t)=0
\end{equation} 
with constant matrices $A_0, \dotsc, A_k \in \mathbb{R}^{n \times n}$ and a smooth function $g: \mathbb{R} \to \mathbb{R}^n$, $\sigma_f(i,j)$ is equal to $\deg_s A(s)_{ij}$, where $A(s) \coloneqq \sum_{l=0}^k A_l s^l$ is the coefficient matrix of the Laplace transformation of~\eqref{lti}.

Structural methods solve the following optimization problem:
\begin{equation*}
   (\mathrm{D}_f):\left\{ \begin{aligned}
        \text{minimize}\quad &\sum_{j \in [n]}q_j -\sum_{i \in [n]}p_i\\
        \text{subject to}\quad &q_j-p_i \geq \sigma_f(i,j), \\
           & p_i,q_j \in \mathbb{Z}_{\ge 0}\quad (\forall i,j \in [n]),
   \end{aligned}\right.
\end{equation*}
which is equivalent to the LP (linear programming) dual of the maximum-weight perfect bipartite matching problem for $G_f$.  We denote the optimal value to $(\mathrm{D}_f)$ by $\hat\delta_f$, where $\hat\delta_f$ is set to be $-\infty$ if $(\mathrm{D}_f)$ is unbounded.
A DAE~$f=0$ is said to be \emph{structurally singular} if $\hat\delta_f=-\infty$ and \emph{structurally nonsingular} otherwise.
The structural singularity means that there is no one-to-one correspondence between the equations and the variables appearing in the DAE.  In this paper, only structurally nonsingular DAEs are treated.

Let $(p,q)$ be an optimal solution to $(\mathrm{D}_f)$.
The \emph{system Jacobian} of a DAE $f=0$ with respect to $(p,q)$ is a functional matrix defined by
\[
  J_f^{p,q}\coloneqq \prn*{\frac{\partial f_i^{(p_i)}} { \partial x_j^{(q_j)}} }_{ij} = \prn*{\frac{\partial f_i }{\partial x_j^{(q_j-p_i)}}}_{ij},
\]
where the last equality comes from Griewank's lemma~\cite[Lemma~3.7]{pryce01}.
The determinant of a system Jacobian does not depend on the choice of $(p,q)$~\cite{murota95}, whereas the system Jacobian itself changes. 

Pryce's algorithm~\cite{pryce01}, which is a reinterpretation of Pantelides' algorithm~\cite{pantelides88}, computes an optimal solution $(p, q)$ to $(\mathrm{D}_f)$ and returns a system of $n+\sum_{i=1}^n p_i$ equations obtained by collecting the zeroth, first, $\dotsc,p_i$th-order derivatives of each equation $f_i=0$.  The new system contains sufficient algebraic constraints of the original DAE in the sense that we can obtain a consistent initial value by solving the new system as a nonlinear equation with respect to $(x_1,\dotsc,x_1^{(q_1)},\dotsc,x_n,\dotsc,x_n^{(q_n)},t)$.
Unfortunately, Pryce's algorithm is inapplicable to DAEs whose system Jacobians are identically singular, i.e., $\det J_f^{p,q}\equiv 0$.
Similarly, the Mattsson--S\"{o}derlind (MS) method~\cite{mattsson93}, a structural index reduction method, also fails for DAEs whose system Jacobian is identically singular.
Such systems can arise from real problems, as we show in the following example.

\begin{figure}[tb]
    \centering
    \includegraphics[width=9cm]{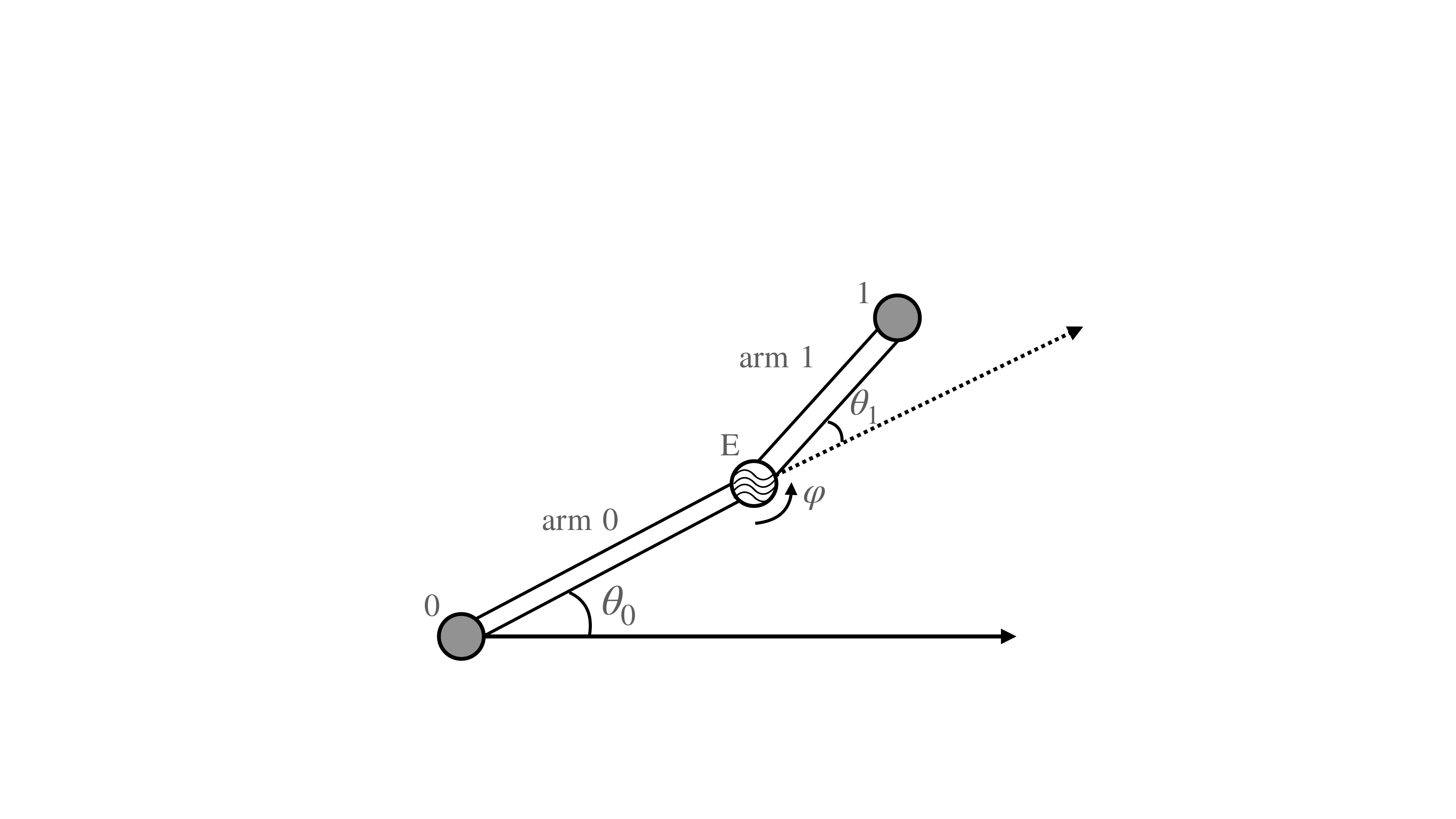}
    \caption{Robotic arm whose behavior is described by~\eqref{roboteq1}.}\label{fig:robotfig}
\end{figure}

\begin{example}[robotic arm~\cite{deluca88}]\label{robotex}
    Consider a robotic arm illustrated in Figure~\ref{fig:robotfig}.
    The joints at the fulcrum and the tip are denoted as joint 0 and 1, respectively, and the joint at the center is named as joint $\mathrm{E}$.  The arms are named as arms $0,1$, sequentially from the fulcrum.
    We denote the angle between the base and arm 0 as $\theta_0$, the angle between arms 0 and 1 as $\theta_1$, the rotation angle of the motor of joint $\mathrm{E}$ as $\varphi$, and the torque on joint $0$ and $\mathrm{E}$ as $\tau_0$ and $\tau_1$, respectively.
    We consider the path control problem, where the objective is to align the horizontal coordinates of joint $\mathrm{E}$ and 1 at time $t$ with some smooth function $p_1(t)$ and $p_2(t)$.
    A DAE expressing the dynamics of $x \coloneqq (\theta_0,\theta_1,\varphi,\tau_0,\tau_1)^\top$ is given by
    \begin{align}\label{roboteq1}
        \left\{\begin{aligned}
            \ddot{x}+\frac{1}{Q}\begin{pNiceMatrix}
                 1 & a         & -1 \\
                 a & Q/\beta+a^2 & a \\
                -1 & a         & Q/J+1
            \end{pNiceMatrix}\begin{pmatrix}
              -(2\dot\theta_0\dot\theta_1-\dot\theta_1^2)\gamma\sin\theta_1-\tau_0\\
              \gamma\dot\theta_0^2\sin\theta_1+K(\theta_1-\varphi/R)\\
              \frac{K}{R} (\varphi / R -\theta_1) - \tau_1
            \end{pmatrix}=0,\\
            l_0 \cos\theta_0=p_1(t),\\
            l_0 \cos\theta_0 + l_1\cos(\theta_0 + \theta_1)=p_2(t),
        \end{aligned}\right.
    \end{align}
    where $a=-(\beta+\gamma\cos \theta_1) / \beta$, $Q=\alpha+2\gamma\cos\theta_1-J-(\beta+\gamma\cos\theta_1)^2 / \beta$, and $\alpha, \beta, \gamma, R, K, l_0, l_1, J$ are non-zero constant parameters; see Section~\ref{sec:experiment-daes} for details.
    The $\sigma$-function of~\eqref{roboteq1} is given by
    \begin{equation}\label{eq:sigmafunc}
    \prn{\sigma_f(i,j)}_{ij}=\begin{pmatrix}
      2&1&0&0&0\\
      1&2&0&0&0\\
      1&1&2&0&0\\
      0&-\infty&-\infty&-\infty&-\infty\\
      0&0&-\infty&-\infty&-\infty
      \end{pmatrix}
    \end{equation}
    and one optimal solution to $(\mathrm{D}_f)$ is $p=(0,0,0,2,2), q=(2,2,2,0,0)$.
    Then, the system Jacobian $J_f^{p,q}$ is the following matrix
    \begin{align}\label{eq:ex-jacobian}
        \begin{pNiceMatrix}
            1& & &-1/Q&1/Q\\
             &1& &1/Q + \gamma \cos \theta_1 / \beta Q & -1/Q - \gamma \cos \theta_1 / \beta Q\\
             & &1&1/Q&-1/Q-1/J\\
            -l_0 \sin \theta_0 & & & & \\
            -l_0 \sin \theta_0- l_1 \sin(\theta_0 + \theta_1)& - l_1 \sin(\theta_0 + \theta_1)& & &
        \end{pNiceMatrix},
    \end{align}
    where empty cells are zero.
    The matrix~\eqref{eq:ex-jacobian} is identically singular since the sum of the fourth and fifth columns equals the third column divided by $-J$.
\end{example}

\section{Regularization Methods}\label{subsec_reg}
Regularization methods aim to transform a DAE with a singular system Jacobian into an equivalent DAE with a nonsingular system Jacobian.
This section describes existing regularization methods.

\subsection{Combinatorial Relaxation Framework}

The \emph{combinatorial relaxation}~\cite{murota90,murota95} is a general framework used in regularization for DAEs~\cite{iwata19,oki23,tan17,wu13}.
The combinatorial relaxation regularizes a DAE~\eqref{dae} through the following three phases.

\begin{enumerate}[{label={\textbf{Phase~\arabic*.}},labelindent=\parindent,leftmargin=*}]
    \item Solve the problem $(\mathrm{D}_f)$. If $\hat{\delta}_f = -\infty$, halt with failure.
    \item If the system Jacobian $J_f^{p,q}$ with respect to an optimal solution $(p, q)$ to $(\mathrm{D}_f)$ is nonsingular, return $f = 0$.
    \item Transform the DAE $f=0$ into an equivalent DAE $f'=0$ with $\hat\delta_{f'}< \hat\delta_f$. Go back to Phase~1.
\end{enumerate}

Since $\hat{\delta}_f \in \mathbb{Z}_{\ge 0}$ for structurally nonsingular DAEs, the above process ends in at most $\hat{\delta}_f \le kn$ iterations.
Phase~1 is executed by the Hungarian method in $O(n^3)$ time.
The nonsingularity testing in Phase~2 can be exactly performed via the symbolic Gaussian elimination on $J^{p,q}_f$, while this operation costs much computational time. Alternatively, we can check the nonsingularity of a constant matrix obtained by substituting random numerical values to symbols in $J^{p,q}_f$, while this approach may lead to an incorrect result due to numerical errors or unlucky substitution.
In the remaining of this section, we explain existing modification methods in Phase~3.

\subsection{Regularizing Linear Time-Invariant DAEs}\label{subsec_lti}

Wu et al.~\cite{wu13} presented a first regularization method for linear time-invariant DAEs~\eqref{lti}, relying on the combinatorial relaxation algorithm by Iwata~\cite{iwata03}.
Here, we explain Phase~3 of a regularization method for~\eqref{lti} based on the algorithm by Murota~\cite{murota95}.
By the Laplace transform, we rewrite~\eqref{lti} as $A(s)\tilde x(s)-\tilde g(s)=0$, where $A(s) \coloneqq \sum_{l=0}^k A_ls^l$, and $\tilde{x}(s)$ and $\tilde{g}(s)$ are the Laplace transformations of $x(t)$ and $g(t)$, respectively.

We introduce additional notions on matrices. 
The \emph{term-rank} of a matrix $A$, denoted by $\trank A$, is the maximum number of non-zero entries that can be chosen from $A$ such that the indices of rows and columns do not overlap.  It can be shown that $\rank A \leq \trank A$ always holds~\cite{edmonds67}.
For a square matrix $A$ in size $n$ and an integer sequence $p=(p_1,\dotsc,p_n)$, we call $A$ \emph{upper-triangular along $p$} if $A_{ij}=0$ holds for all pairs $i,j$ such that $p_i > p_j$.

Let $(p, q)$ be an optimal solution to $(\mathrm{D}_f)$ computed at Phase~2; the system Jacobian $J = J_f^{p,q}$ of~\eqref{lti} is a constant matrix.
By the Gaussian elimination, Murota's method finds a nonsingular matrix $U\in \mathbb{R}^{n \times n}$ such that $U$ is upper-triangular along $p$ and $\trank UJ < n$ holds.
Such $U$ exists if and only if $J$ is singular, i.e., $U$ is a ``certificate of singularity'' of $J$ in this sense.
Then, the method transforms $f = 0$ into a DAE $f' = 0$ whose Laplace transformation is given by
\begin{equation}\label{unimodtrans}
  U(s)A(s)\tilde{x}(s) - U(s)\tilde{g}(s) = 0,
\end{equation}
where $U(s) \coloneqq \diag\{s^{-p_i}\}U\diag\{s^{p_i}\}$.

Since $U(s)$ is a \emph{unimodular matrix}, i.e., a polynomial matrix whose inverse is also a polynomial matrix, the multiplication by $U(s)$ is an invertible operation.
That is, the solution sets of $f=0$ and $f'=0$ are identical.
Moreover, $\hat\delta_{f'}<\hat\delta_f$ holds as desired~\cite{murota90}; this essentially follows from the condition $\trank UJ <n$.

\subsection{Regularizing Nonlinear DAEs}
The above regularization method has been extended to nonlinear DAEs~\cite{iwata19,oki23,tan17}.
The LC-method by Tan et al.~\cite{tan17} finds a non-zero symbolic vector $v$ belonging to the left kernel of the system Jacobian $J$ and then transforms the DAE using $v$ in a way analogous to the transformation given in~\eqref{unimodtrans}, where $s$ in $U(s)$ is now regarded as the differential operator $D_t$.
Tan et al.~\cite{tan17} also proposes another method, the ES-method, which uses a non-zero vector in the right kernel of $J$ and modifies the DAE by change of variables.
However, the system Jacobian of a nonlinear DAE is no longer a constant matrix, requiring time-consuming symbolic computations for elimination to obtain a symbolic non-zero vector in the kernels.
Oki's substitution and augmentation methods~\cite{oki23} use a submatrix of $J$ with certain properties as a certificate of singularity, and these methods can be applied to DAEs in which the availability condition of the LC- or ES-method is not satisfied.
This certificate is also found via symbolic elimination, as implemented in~\cite{git}.

We here explain the Iwata--Oki--Takamatsu (IOT) method~\cite{iwata19}, which regularizes nonlinear DAEs without symbolic elimination.
Their idea is to ``approximate'' $J$ by a simpler symbolic matrix, called a \emph{layered mixed matrix}.
The IOT method assumes that a DAE is given as the sum of linear time-invariant terms and other terms as
\begin{equation}\label{half}
        \sum_{l=0}^k A_lx^{(l)}(t)+g(x(t),\dot{x}(t),\dotsc,x^{(k)}(t),t)=0,
\end{equation}
where $A_0, \dotsc, A_k$ are $n \times n$ constant matrices and $g$ is a smooth function; note that such a decomposition is not uniquely determined.

The IOT method first converts~\eqref{half} into an equivalent DAE 
\begin{equation}
    \left\{ \begin{aligned}\label{layer}
      y(t)+\sum_{l=0}^k A_lx^{(l)}(t)=0,\\
      -y(t)+g(x(t),\dot{x}(t),\dotsc,x^{(k)}(t),t)=0,
    \end{aligned} \right.
\end{equation}
where $(x(t), y(t))$ is a vector variable of dimension $2n$.
We call~\eqref{layer} the \emph{layer form} of~\eqref{half}.
The top $n$ rows of a system Jacobian $J$ of~\eqref{layer} consist of constants and the bottom $n$ rows may contain functional entries.
The method converts $J$ into a simpler symbolic matrix $J_\mathrm{M}$ by replacing the non-zero entries in the bottom $n$ rows of $J$ with distinct symbols $\alpha_1,\dotsc,\alpha_m$.
The resulting matrix $J_\mathrm{M}$ is a \emph{layered mixed matrix} defined as follows.

\begin{definition}[layered mixed matrix~\cite{murota85}]
    A \emph{layered mixed matrix} is a matrix $A$ over a rational function field $\mathbb{R}(\alpha_1, \dotsc, \alpha_m)$ in the form of $A = \binom{Q}{T}$ such that (i) each entry in $Q$ belongs to $\mathbb{R}$ and (ii) each nonzero entry in $T$ is a distinct symbol $\alpha_i$.
\end{definition}

The rank of a layered mixed matrix can be deterministically calculated without symbolic computations by solving a combinatorial optimization problem, called the \emph{independent matching problem}~\cite{murotatext}, which is a minor extension of the \emph{matroid intersection problem}.
By Cunninghum's algorithm~\cite{cunningham83}, we can solve the independent matching problem in $O(n^3\log n)$ operations over $\mathbb{R}$.
Furthermore, the following rank formula holds for layered mixed matrices, which follows from the min-max theorem of independent matching.

\begin{theorem}[{\cite[Theorem 4.2.3]{murotatext}}]\label{rankofmixed}
    Let $A=\binom{Q}{T}$ be a layered mixed matrix with columns $C$ and $m_Q$ be the number of rows of $Q$.
    Then, we have
    \begin{equation}\label{lmrankdual}
        \rank A = \min_{I \subseteq C,|I|=m_Q}\{\rank Q[I]+ \trank T[I]-|C\setminus I|\}.
    \end{equation}
\end{theorem}

A column set $I$ attaining the minimum of~\eqref{lmrankdual} for a layered mixed matrix $A 
= \binom{Q}{T}$ can also be obtained as a byproduct of Cunninghum's algorithm.
Eliminating a constant submatrix $Q[I]$, we can construct a nonsingular constant matrix $U$ such that $\rank A = \trank UA$ holds without symbolic operations; see~\cite{iwata19}. 

Since $J_\mathrm{M}$ is obtained from $J$ by discarding algebraic relationships between non-zero entries in the bottom $n$ rows, we always have $\rank J \le \rank J_\mathrm{M}$.
In particular, if $J_\mathrm{M}$ is singular, so is $J$.
Furthermore, a certificate $U \in \mathbb{R}^{2n \times 2n}$ of the singularity of $J_\mathrm{M}$ is also that of $J$, i.e., $\trank UJ = \trank UJ_\mathrm{M} < n$ holds.
Therefore, we can transform the DAE~\eqref{layer} into an equivalent DAE $f' = 0$ with $\hat{\delta}_{f'} < \hat{\delta}_f$ using $U$ in the same manner explained in Section~\ref{subsec_lti}.

Unfortunately, the singularity of $J$ does not necessarily imply that of $J_\mathrm{M}$, as the following example demonstrates.

\begin{example}\label{ex:iot-failure}
    Let us consider a DAE
    \begin{align}\label{exeq}
        f \coloneqq \begin{pmatrix}
            0 \\
            x_2 \\
            \dot{x}_2 + \dot{x}_3
        \end{pmatrix}
        + \begin{pmatrix}
            \dot{x_1}\dot{x_2}+\dot{x_2}\dot{x_3}+\dot{x_3}\dot{x_1}+\dot{x_3}^2 \\
            \cos(\dot{x_1}-\dot{x_2}) \\
            \cos(\dot{x_1}-\dot{x_2})
        \end{pmatrix}
        = 0.
    \end{align}
    The layer form~\eqref{layer} of~\eqref{exeq} has an optimal solution $p = (0,0,0,0,0,0)$, $q = (1,1,1,0,0,0)$ and the associated system Jacobian is
    \begin{equation}\label{exeq_jac}
        \begin{pNiceMatrix}
            & &   & 1 &   &   \\
            & &   &   & 1 &   \\
            & 1 & 1 &   &   & 1 \\
            \dot{x_1}+\dot{x_3}& \dot{x_2}+\dot{x_3}&\dot{x_1}+\dot{x_2}+2\dot{x_3}&-1& &\\
            -\sin(\dot{x_1}-\dot{x_2})&\sin(\dot{x_1}-\dot{x_2})&&&-1& \\
            -\sin(\dot{x_1}-\dot{x_2})&\sin(\dot{x_1}-\dot{x_2})&&&&-1 
        \end{pNiceMatrix},
    \end{equation}
    which is singular.
    However, the IOT method approximates~\eqref{exeq_jac} by
    \begin{equation}\label{exeq_js}
        J_\mathrm{M}=\begin{pmatrix}
            & &   & 1 &   &   \\
            & &   &   & 1 &   \\
            & 1 & 1 &   &   & 1 \\
            \alpha_1&\alpha_2&\alpha_3 & -1& &\\
            \alpha_4&\alpha_5& & &-1& \\
            \alpha_6&\alpha_7& & & &-1
        \end{pmatrix},
    \end{equation}
    which is nonsingular.
\end{example}

The IOT method fails for~\eqref{exeq} because it replaces the $(5,1)$, $(5,2)$, $(6,1)$, $(6,2)$ entries of~\eqref{exeq_jac} with distinct symbols $\alpha_4, \dotsc, \alpha_7$, although they were indeed like terms $\pm \sin(\dot{y_1}-\dot{y_2})$.
This is a limitation of the IOT method that relies on layered mixed matrices, in which each symbol can occur only once.

\begin{remark}\label{rem:physical-param}
    The IOT method was initially proposed for regularizing linear time-invariant DAEs that contain physical parameters. In the spirit of mixed matrices~\cite{murotatext}, these parameters are modeled as independent symbols due to tolerance or measurement errors.
    Even in the application to nonlinear DAEs explained in this section, we can also treat terms of physical parameters as functional terms, even though they are indeed constants.
\end{remark}

\section{Proposed Method}\label{sec_prp}

In this section, to capture algebraic relationships in more detail, we propose a regularization method that approximates the system Jacobian by a more expressive but still computationally tractable symbolic matrix.
We first describe an overview in Section~\ref{subsec_prp_overview} and then explain each step of the method in Sections~\ref{subsec_prp_1CM-matrix}--\ref{subsec_prp_trank}.

\subsection{Overview}\label{subsec_prp_overview}

Let $f=0$ be a DAE~\eqref{dae}, where there is no need to assume the decomposition~\eqref{half}. 
Let $(p, q)$ an optimal solution to $(\mathrm{D}_f)$.
Our method first approximates the system Jacobian $J = J_f^{p,q}$ by a \emph{linear symbolic matrix} $J_\mathrm{LS} = A_0 + \alpha_1 A_1 + \dotsb + \alpha_m A_m$ defined as follows.

\begin{definition}
    A \emph{linear symbolic matrix} is a rational function matrix $A\in \mathbb{R}(\alpha_1,\dotsc,\alpha_m)^{n\times n}$ such that every entry is a polynomial in $\alpha_1,\dotsc,\alpha_m$ of degree at most one.  A linear symbolic matrix is expressed as
    $
        A=A_0+\alpha_1 A_1+\dotsb+\alpha_mA_m,
    $
    where $A_0,\dotsc,A_m \in \mathbb{R}^{n\times n}$.
\end{definition}

As with the approximation by a layered mixed matrix $J_\mathrm{M}$, $\rank J \le \rank J_\mathrm{LS}$ always holds, where the inequality can be strict.
Nevertheless, $J_\mathrm{LS}$ can be a ``better'' approximation than $J_\mathrm{M}$, as the following shows.

\begin{example}\label{ex:toy}
    Consider the system Jacobian~\eqref{exeq_jac} of the layer form of the DAE~\eqref{exeq} again.
    Without minding multiple occurrences of the same symbols, we approximate~\eqref{exeq_jac} by
    \begin{align}\label{lsm-ex}
        J_\mathrm{LS}=\begin{pNiceMatrix}
            & &   & 1 &   &   \\
            & &   &   & 1 &   \\
            & 1 & 1 &   &   & 1 \\
            \alpha_1+\alpha_3&\alpha_2+\alpha_3&\alpha_1+\alpha_2+2\alpha_3 & -1& &\\
            -\alpha_4&\alpha_4& & &-1&\\
            -\alpha_4&\alpha_4& & & &-1
        \end{pNiceMatrix},
    \end{align}
    where we mapped as $\dot{x_1}\mapsto \alpha_1, \dot{x_2}\mapsto \alpha_2, \dot{x_3}\mapsto \alpha_3,$ and $\sin(\dot{x_1}-\dot{x_2}) \mapsto \alpha_4$. The resulting matrix $J_\mathrm{LS}$ is a linear symbolic matrix.
    In contrast to $J_\mathrm{M}$ given in~\eqref{exeq_js}, the matrix $J_\mathrm{LS}$ is still singular, as the sum of the first and second columns are equal to the last column.
\end{example}

\begin{example}[{Continued from Example~\ref{robotex}}]\label{ex:robotex-lsm}
    For the system Jacobian~\eqref{eq:ex-jacobian} of the DAE~\eqref{roboteq1} describing path control of a robotic arm, we can construct the following linear symbolic matrix
    \begin{align}\label{eq:robot-ls}
        J_\mathrm{LS} = \begin{pmatrix}
            1& & &-\alpha_1 &\alpha_1\\
             &1& &\alpha_1 + \alpha_2 & -\alpha_1 - \alpha_2\\
             & &1&\alpha_1&-\alpha_1-\alpha_3\\
            -\alpha_4& & & & \\
            -\alpha_4- \alpha_5& - \alpha_5& & &
        \end{pmatrix},
    \end{align}
    where we replaced terms as $-1/Q \mapsto \alpha_1$, $\gamma c_1 / \beta Q \mapsto \alpha_2$, $1/J \mapsto \alpha_3$, $l_0 s_0 \mapsto \alpha_4$, and $l_1 s_{01} \mapsto \alpha_5$.
    The matrix~\eqref{eq:robot-ls} is still singular.
\end{example}

Nonsingularity testing for linear symbolic matrices is called \emph{Edmonds' problem}~\cite{edmonds67}.
Unfortunately, it is a long-standing open problem whether one can solve Edmonds' problem in deterministic polynomial time; hence performing the combinatorial relaxation using linear symbolic matrices is still a difficult task.

To alleviate this problem, we further exploit the ``low-rank structure'' of the coefficient matrices $A_1, \dotsc, A_m$ in $J_\mathrm{LS} = A_0 + \alpha_1 A_1 + \dotsb + \alpha_m A_m$ arising from real instances.
Specifically, the coefficient matrices $A_k$, except the constant term $A_0$, tend to have rank one, i.e., $\rank A_k = 1$ for all $k \geq 1$.
We call such a linear symbolic matrix a \emph{rank-1 coefficient mixed matrix} (1CM-matrix); the matrices~\eqref{lsm-ex} and~\eqref{eq:robot-ls} are actually 1CM-matrices.
In Appendix~\ref{app-1cm}, we show that DAEs represented via some ``composition'' of linear combination of variables and nonlinear functions have system Jacobians which can be approximated by 1CM-matrices.  This class includes DAEs of electrical circuits modeled by modified nodal analysis~\cite{ho75}, or multi-body mechanical systems.

Even if some of the ranks of the coefficient matrices are grater than one, we can construct a linear symbolic matrix $J_\mathrm{LS}$ from a 1CM-matrix $J_\mathrm{1CM}$ with $\rank J_\mathrm{1CM} \ge \rank J_\mathrm{LS}$ by using rank decomposition, although this decomposition may lose the singularity of $J_\mathrm{LS}$.

Ivanyos et al.~\cite{ivanyos10} showed that for any singular 1CM-matrix $A = A_0 + \alpha_1 A_1 + \dotsb + \alpha_m A_m$, there exist nonsingular matrices $U, V \in \mathbb{R}^{n \times n}$ such that $\trank UAV < n$ holds.
We call such $(U, V)$ a \emph{vanishing pair} of $A$; it serves as a certificate of singularity of $A$.
Ivanyos et al.~\cite{ivanyos10} also proposed an algorithm to find a vanishing pair of a 1CM-matrix that runs in $O(n^5m)$ operations over $\mathbb{R}$.\footnote{For the special case with $A_0=O$, a vanishing pair can be found by solving the dual problem to linear matroid intersection~\cite{lovasz89}.}
We propose a faster algorithm to find a vanishing pair of a 1CM-matrix that runs in $O((n+m)^3\log(n+m))$ time, which is faster than the algorithm by Ivanyos et al. because we can assume $m \le n^2$ by picking up a subset of $A_1,\dotsc,A_m$ which is linearly independent in the linear space of $n\times n $ matrices.  This can be done by elininating a $m\times n^2$ matrix and can be done in $O(\min\{mn^4,m^2n^2\})$ time, which is still faster than $O(n^5m)$.

A vanishing pair of $J_\mathrm{1CM}$ is that of $J$ as well.
Using it, we can convert $f = 0$ into an equivalent DAE $f' = 0$ with $\hat{\delta}_{f'}< \hat{\delta}_{f}$.

In summary, our method performs Phase~3 of the combinatorial relaxation in the following three steps.

\begin{enumerate}[{label={\textbf{Step~\arabic*.}},labelindent=\parindent,leftmargin=*}]
    \item Approximate the system Jacobian $J_f^{p,q}$ of the DAE $f = 0$ by a 1CM-matrix $J_\mathrm{1CM}$.
    \item Halt if $J_\mathrm{1CM}$ is nonsingular. Otherwise, find a vanishing pair $(U, V)$ of $J_\mathrm{1CM}$.
    \item Transform the DAE~$f=0$ into an equivalent DAE~$f'=0$ with $\hat\delta_{f'}<\hat\delta_f$ using $(U, V)$.
\end{enumerate}

We explain details of Steps 1--3 in Sections~\ref{subsec_prp_1CM-matrix}--\ref{subsec_prp_trank}, respectively.

\begin{remark}
  While $J_\mathrm{1CM}$ can capture more algebraic dependencies of entries in a system Jacobian $J = J_f^{p,q}$ than $J_\mathrm{M}$, it can still happen that $n = \rank J_\mathrm{1CM} > \rank J$.
  In this case, our method outputs a DAE with a singular system Jacobian like the IOT method does, keeping the structural methods inapplicable.
  Thus, for such DAEs, we need to switch to a slower but a more applicable regularization method, as mentioned in Section~\ref{sec:introduction}.
  
\end{remark}

\subsection{Step 1: Constructing 1CM-matrix}\label{subsec_prp_1CM-matrix}

This section presents a heuristic approach to approximate a system Jacobian $J$ by a 1CM-matrix $J_\mathrm{1CM}$.
We first construct an intermediate linear symbolic matrix $J_\mathrm{LS}$, which is not necessarily a 1CM-matrix.
Then, we convert it into a 1CM-matrix $J_\mathrm{1CM}$.

\subsubsection*{From System Jacobian to Linear Symbolic Matrix}
We expand $J$ as $J=A_0+h_1A_1+\cdots+h_mA_m$, where each $h_i$ is a function and $A_i$ is a constant matrix, where we allow $h_i = h_j$ for $i \ne j$.
Then, replacing each $h_i$ with a symbol $\alpha_i$, we obtain a linear symbolic matrix $J_\mathrm{LS}=A_0+\alpha_1 A_1+\cdots+\alpha_mA_m$.
It is easy to verify that a vanishing pair of $J_\mathrm{LS}$ is that of $J$, while the converse does not necessarily hold.
Thus, it is desirable to obtain an expansion such that the existence of a vanishing pair of $J$ implies the existence of that of $J_\mathrm{LS}$.
If we could obtain linearly independent $h_1, \dotsc, h_m$, i.e., $c_1 h_1 + \dotsb + c_m h_m \not\equiv 0$ for all $(c_1, \dotsc, c_m) \ne 0$, then any vanishing pair of $J$ is that of $J_\mathrm{LS}$.
However, unfortunately, even determining whether a set of functions is linearly independent is an undecidable problem in general, because even determining whether some expression is nonzero is undecidable in general~\cite{richardson69}, and thus computing such a decomposition is impossible.  We here present one heuristic algorithm to obtain a decomposition, avoiding linear dependencies among $h_1, \dotsc, h_m$ as much as possible.

First, we decompose each $(i,j)$ entry $J_{ij}$ of $J$ as $J_{ij}=\sum_l c_{ijl}g_{ijl}$, where $c_{ijl}$ is a constant and $g_{ijl}$ is a function.  There are many ways for this decomposition, but in general, one should decompose $J_{ij}$ as many terms as possible, except that evidently like terms are merged. For example, if $J_{ij}$ is a polynomial in $x_1,\dotsc,x_1^{(k)},\dotsc,x_n,\dotsc,x_n^{(k)}$, then one should expand the polynomial and set each monic monomial term as $g_{ijl}$ and its coefficient as $c_{ijl}$.  This is because the set of distinct monic monomials is linearly independent.  For a non-polynomial entry, one simple and effective method is to decompose $J_{ij}$ ``arithmetically'': we expand $(f+g)h$ as $fh+gh$ for functions $f,g,h$ but do not transform non-polynomial functions such as $\sin(x+y)$.
For instance, we can perform such an expansion by using \texttt{expand} function in Symbolic Math Toolbox of MATLAB with setting \texttt{arithmeticOnly} option as \texttt{true}\footnote{\url{https://www.mathworks.com/help/symbolic/choose-function-to-rearrange-expression.html} (accessed February 1, 2024).}.

Next, we find identical functions among $\{g_{ijl}\}_{i,j,l}$.  This procedure, however, is again undecidable in general~\cite{richardson69}.  Thus, we use the following heuristics.
\begin{description}[{leftmargin=\parindent,labelindent=\parindent}]
    \item[\textbf{Hash}:] generate a hash value $H_f$ from a function $f$. Different functions must have different hash values, i.e., $H_f\neq H_g$ for $f\not\equiv g$. The same functions are expected to have the same hash values, i.e., $H_f=H_g$ for $f\equiv g$.
\end{description}
Note that usual hash functions used in data structures are required to hold $H_f=H_g$ if $f \equiv g$ and is expected to hold $H_f\neq H_g$ if $f\not\equiv g$; the opposite nature is required for our purpose.  Such heuristics can be constructed, for example, by applying a procedure to simplify a function representation (for instance, \texttt{simplify} function in MATLAB), and then hashing the string representation of the simplified representation of the function.

Finally, we identify $g_{ijl}$ and $g_{i'j'l'}$ if their hash values are the same.  After the identification, we rename the functional terms as $h_1,\dotsc,h_m$ and obtain an expansion $J=A_0+h_1A_1+\cdots+h_mA_m$.

\subsubsection*{From Linear Symbolic Matrix to 1CM-matrix}

Let $J_\mathrm{LS}=A_0+\alpha_1 A_1+\cdots+\alpha_m A_m$ be a linear symbolic matrix and $r(i)=\rank A_i$ for $i \in [m]$.  It is well-known that each $A_i$ can be decomposed into the sum of $r(i)$ matrices whose ranks are one as $A_i=A_{i1}+\cdots+A_{ir(i)}$.
Using this decomposition, we construct a 1CM-matrix $J_\mathrm{1CM}$ as
\begin{equation}\label{rb}
    \begin{aligned}
        J_\mathrm{1CM}=A_0+&\alpha_{11}A_{11}+\cdots+\alpha_{1r(1)}A_{1r(1)}+ \dotsb \\
        &\quad+\alpha_{m1}A_{m1}+\cdots+\alpha_{mr(m)}A_{mr(m)},
    \end{aligned}
\end{equation}
where $\alpha_{ij}$ are new symbols.
It is clear that any vanishing pair $(U,V)$ of $J_\mathrm{1CM}$ is that of $J_\mathrm{LS}$, hence that of $J$.

\subsection{Step 2: Finding Vanishing Pair}\label{ncdual}

This section proposes a fast algorithm to find a vanishing pair $(U, V)$ of a 1CM-matrix without using symbolic computation.
A key property is the min-max formula for the ranks of 1CM-matrices given by Soma~\cite{soma14}, which extends Theorem~\ref{rankofmixed}.
Let $A = A_0 + \alpha_1 A_1 + \dotsc + \alpha_m A_m$ be an $n \times n$ 1CM-matrix.
Using $n$-dimensional non-zero vectors $b_1,\dotsc,b_m,c_1,\dotsc,c_m \in \mathbb{R}^n$, we can write $A$ as
\begin{equation}\label{vec}
    A=A_0+\alpha_1 b_1c_1^\top+\cdots+\alpha_m b_mc_m^\top.
\end{equation}
We call~\eqref{vec} a \emph{vector representation} of $A$.
    
\begin{theorem}[{\cite[Theorem 3.4]{soma14}}]\label{thm:1cm-minmax}
    Consider a vector representation~\eqref{vec} of a 1CM-matrix $A$ and let $B = (b_1 \dotsb b_m)$ and $C = (c_1 \dotsb c_m)$.
    Then, it holds that
    \begin{equation}\label{1cmrank}
    \rank A=\min  \Set*{\rank \begin{pmatrix}A_0 & B[I]\\C[[m]\setminus I]^\top &O\end{pmatrix}}{I \subseteq [m]}.
    \end{equation}
\end{theorem}

Cunningham's algorithm~\cite{cunningham83} can also be used to find a minimizer $I$ of~\eqref{1cmrank}, whereas this algorithm does not yield a vanishing pair of $A$.
Indeed, the proof of Theorem~\ref{thm:1cm-minmax} in~\cite{soma14} is based on the rank formula of layered mixed matrices obtained by expanding a 1CM-matrix, and the auxiliary graph constructed in Cunningham's algorithm to calculate $\rank A$ and optimal $I$ does not provide information on a vanishing pair of the original 1CM-matrix.

\begin{algorithm}[tb]
    \caption{Algorithm to find a vanishing pair of a 1CM-matrix}\label{alg_main}
    \begin{algorithmic}[1] 
        \Require 1CM-matrix $A = A_0 + \alpha_1 b_1 c_1^\top + \dotsb + \alpha_m b_m c_m^\top$
        \Ensure Vanishing pair $(U, V)$ of $A$
        \State $B \gets (b_1 \dotsb b_m)$, $C \gets (c_1 \dotsb c_m)$
        \State Find a minimizer $I$ of~\eqref{1cmrank} by Cunningham's algorithm
        \If{$\rank A = n$}
            \State Report ``$A$ is nonsingular'' and halt
        \EndIf
        \State Find nonsingular $U\in \mathbb{R}^{n\times n}$ such that $n-\rank B[I]$ rows of $UB[I]$ are zero. Let $S$ be the zero rows.
        \State Find nonsingular $V\in \mathbb{R}^{n\times n}$ such that $n-\rank C[[m]\setminus I]^\top$ columns of $C[[m]\setminus I]^\top V$ are zero. Let $T$ be the zero columns.
        \State Find nonsingular $W\in \mathbb{R}^{n\times n}$ such that $|S|-\rank (UA_0V)[S,T]$ rows of $(WUA_0V)[S,T]$ are zero.
        \State $U\leftarrow WU$
        \State \Return $(U, V)$
    \end{algorithmic}
\end{algorithm}

Algorithm~\ref{alg_main} describes our algorithm to find a vanishing pair.
Our idea for Algorithm~\ref{alg_main} stems from the following observation: by Kőnig's theorem~\cite{konig31}, $\trank UAV < n$ is equivalent to the statement that $UAV$ has an all-zero submatrix whose row and column sizes sum to at least $n+1$. 
If $A_0 = O$, eliminating $B[I]$ and $C[[m] \setminus I]^\top$ yields such a submatrix in $A$, see, e.g.,~\cite[Theorem~2.2]{furue20}.
In our setting, however, the constant term $A_0$ hinders the submatrix from being all-zero.
We thus further eliminate $A_0$ carefully not to break the zero submatrix produced by the eliminations of $B[I]$ and $C[[m] \setminus I]^\top$.
The validity and the running time of Algorithm~\ref{alg_main} are summarized as follows.

\begin{theorem}\label{main_1cmdual}
    For a singular $n \times n$ 1CM-matrix $A = A_0 + \alpha_1 A_1 + \dotsb + \alpha_m A_m$, Algorithm \ref{alg_main} outputs a vanishing pair $(U, V)$ of $A$ in $O({(n+m)}^3\log(n+m))$ arithmetic operations on $\mathbb{R}$.
\end{theorem}

Particularly, we can obtain a vanishing pair $(U, V)$ such that $U$ and $V$ are upper-triangular along $p \in \mathbb{Z}^n$ and $q \in \mathbb{Z}^n$, respectively, by eliminating rows (resp.\ columns) having smaller $p$ (resp.\ $q$) values at Lines~5--7 in Algorithm~\ref{alg_main}.
Such $(U, V)$, we call \emph{upper-triangular along} $(p, q)$, will be used in Step~3.

Before presenting the proof of Theorem~\ref{main_1cmdual}, we introduce additional representations for 1CM-matrices, which play an important role in our proof.
Let $A = A_0 + \alpha_1 A_1 + \dotsc + \alpha_m A_m$ be an $n \times n$ 1CM-matrix and consider its vector representation~\eqref{vec}. 
We define the \emph{sparse representation}~\cite{soma14} of $A$ as the following $(n + m) \times (n + m)$ matrix
    \begin{equation}\label{sp}
        \tilde A=\begin{pNiceArray}[margin,vlines=2,hlines=2]{c|ccc}A_0&b_1&\Cdots&b_m\\c_1^\top&\alpha_1\\\Vdots&&\Ddots\\c_m^\top&&&\alpha_m\end{pNiceArray},
    \end{equation}
    and the \emph{layered sparse representation} as the following $(n + 2m) \times (n + 2m)$ matrix
    \begin{equation}\label{lsp}
         \tilde {\tilde A}=\begin{pNiceArray}[margin]{ccc|ccc|ccc}
             1& & &\Block{3-3}{O} & & & &c_1^\top& \\
             &\Ddots& & & & & &\Vdots& \\
             & &1& & & & &c_m^\top& \\\hline
             \Block{3-3}{O}&  & & & & &\Block{3-3}{A_0} & & \\
             &&&b_1 &\Cdots&b_m& & & \\
             & & & & & & & & \\\hline
             \alpha_1& & &\alpha_{m+1}& & &\Block{3-3}{O} & & \\
             &\Ddots& & &\Ddots& & && \\
             & & \alpha_m& & & \alpha_{2m}& & &
         \end{pNiceArray}.
    \end{equation}

Note that $\tilde {\tilde A}$ is a layered mixed matrix.  It is known that there is the following relationship between the ranks of a 1CM-matrix and its (layered) sparse representation.
\begin{proposition}[{\cite[Lemma 3.1]{soma14}}]
For a 1CM-matrix $A$,
\begin{equation*}
    \rank A=\rank {\tilde A}-m=\rank\tilde {\tilde A}-2m
\end{equation*}
holds.
\end{proposition}

It is shown in~\cite{soma14} that $I$ attaining the minimum of the right-hand side of~\eqref{1cmrank} can be obtained as follows: first, find $I' \subseteq [3m]$ attaining the minimum in the formula~\eqref{lmrankdual} for $\tilde {\tilde A}$.  Then,
\begin{equation}
    I=\Set*{i \in [m]}{i,i+n \in I'}.
\end{equation}
minimizes the right-hand side of~\eqref{1cmrank}.

Now we give the proof of Theorem~\ref{main_1cmdual}.

\begin{proof}[{Proof of Theorem~\ref{main_1cmdual}}]
From the equation~\eqref{1cmrank},
\begin{equation*}
    \rank J_\mathrm{1CM}=\rank \begin{pmatrix}
     A_0 & B[I]\\C[[m]\setminus I]^\top&O
\end{pmatrix}
\end{equation*}
holds.  Now we define a matrix $M$ as

\begin{equation*}
    M=\begin{pmatrix}
        WU&O\\O&I_m
    \end{pmatrix}\begin{pmatrix}
     A_0 & B[I]\\C[[m]\setminus I]^\top&O
\end{pmatrix}\begin{pmatrix}
        V&O\\O&I_m
    \end{pmatrix}.
\end{equation*}
Then, $M$ has the form shown in Figure~\ref{sp1}.
\begin{figure}
    \centering
    \includegraphics[width=9cm]{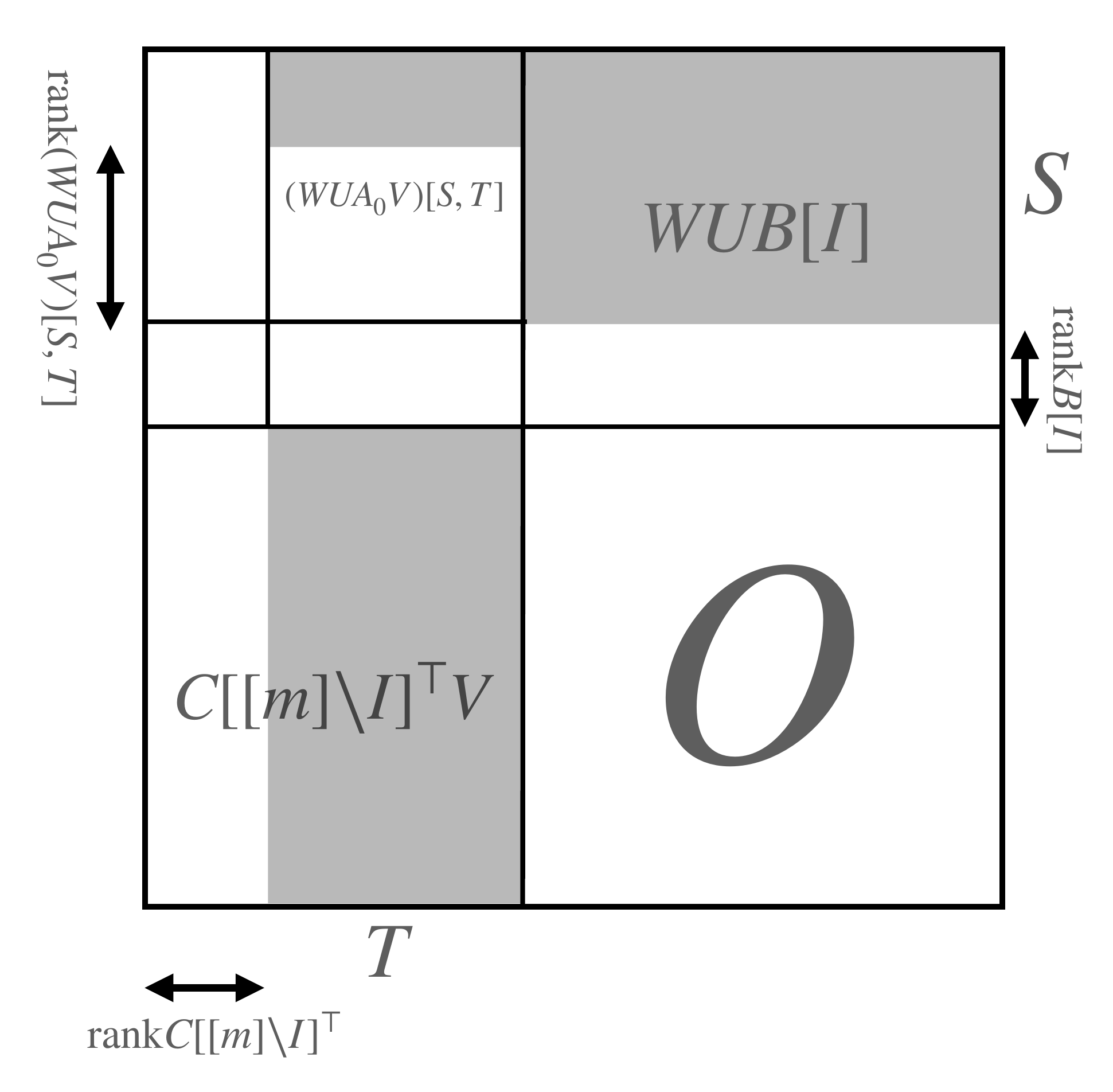}
    \caption{Structure of the matrix $M$.  Grayed out blocks represent zero matrices.}\label{sp1}
\end{figure}
We can easily obtain
\begin{align*}
\rank J_\mathrm{1CM}&=\rank M\\
&= \rank B[I]+\rank C[[m]\setminus I]^\top+\rank (WUA_0V)[S,T].
\end{align*}
On the other hand, noting that $WUJ_\mathrm{1CM}V$ can be represented as
\begin{equation}
WUJ_\mathrm{1CM}V=WUA_0V+\alpha_1 (WUb_1)(c_1^\top V)+\cdots+\alpha_m (WUb_m)(c_m^\top V).
\end{equation}
If we write $WUJ_\mathrm{1CM}V$ as $WUJ_\mathrm{1CM}V=A'_0+\alpha_1 A'_1+\cdots+\alpha_m A'_m$, all of the following submatrices are zero:

\begin{itemize}
\item the $n-\rank B[I]$ rows corresponding to $S$ in $A'_i ~(i \in I)$,
\item the $n -\rank C[[m]\setminus I]^\top$ columns corresponding to $T$ in $A'_i ~(i \not\in I)$, and
\item the submatrix of $A'_0[S,T]$ in $n-\rank B[I]-\rank (WUA_0V)[S,T]$ rows and $n-\rank C[[m]\setminus I]^\top$ columns. 
\end{itemize}

Therefore, in the whole, there is a zero block whose size (the sum of the row and column sizes) is
\begin{align*}
    2n-\rank B[I]-\rank C[[m]\setminus I]^\top-\rank (WUA_0V)[S,T]= 2n-\rank J_\mathrm{1CM},
\end{align*}
and this value is more than $n$ by the assumption.  That is, this zero block is a Hall blocker, and thus $\trank WUJ_\mathrm{1CM}V<n$ holds: $(WU,V)$ is a vanishing pair of $J_\mathrm{1CM}$.

Finally, we perform complexity analysis. First, the problem of finding $I$ can be performed in $O((n+m)^3\log(n+m))$ time using Cunningham's algorithm.  Finding $U,V$ can be realized by Gaussian elimination in $O(mn^2)$ time and it takes $O(n^3)$ time to find $W$.  In summary, it follows that the overall computational complexity is $O((n+m)^3\log(n+m))$. \qedhere
\end{proof}

\subsection{Step 3: Modifying DAE Using Vanishing Pair}\label{subsec_prp_trank}

This section explains Step~3 in our regularization method.
Let $(U, V)$ be a vanishing pair of a system Jacobian $J$ that is upper-triangular along an optimal solution $(p, q)$ to $(\mathrm{D}_f)$.
Let $D_t = \frac{\mathrm{d}}{\mathrm{d}t}$ be the differential operator, which maps a smooth function $z(t)$ to $D_t z(t) \coloneqq \dot{z}(t)$.
We define $U(D_t)\coloneqq\diag\{D_t^{-p_i}\}U\diag\{D_t^{p_i}\}$ and $V(D_t)\coloneqq\diag\{D_t^{-q_i}\}V\diag\{D_t^{q_i}\}$.  They are $n\times n$ matrices including differential operator.
Thanks to the upper-triangularity of $(U, V)$, the degree in $D_t$ of non-zero entries of $U(D_t)$ and $V(D_t)$ are nonnegative.

Our method first transforms the DAE $f = 0$ by left-multiplying $f$ by $U(D_t)$, analogously to~\eqref{unimodtrans}.
We further modify the DAE by changing the variable $x(t)$ into $y(t)$ as $x(t) = V(D_t)y(t)$.
Namely, the resulting system $f' = 0$ is a DAE in $y(t)$ expressed as
\begin{equation}
    f'(y, \dot{y}, \dotsc, t) \coloneqq U(D_t) f({V(D_t)}y, {V(D_t)}\dot{y}, \dotsc, t).
\end{equation}
Intuitively, the transformations by $U(D_t)$ and $V(D_t)$ correspond to the ``row and column operations'' by $U$ and $V$ on $f$, respectively, and provably reduce the value of $\hat\delta_f$.

\begin{theorem}\label{main}
    For the DAEs $f=0$ and $f'=0$, we have $\hat\delta_{f'}<\hat\delta_f$.
\end{theorem}

\begin{proof}
    Let $(p,q)$ be any optimal solution to $(\mathrm{D}_f)$, $J_f^{p,q}$ be the system Jacobian, and let $g=U(D_t)f.$  Note that
    \begin{equation}\label{eq:rowtrans}
        g_i \coloneqq (U(D_t)f)_i= \sum_{k=1}^n U_{ik}f_k^{(p_k-p_i)} \quad (i \in [n])
    \end{equation}
    and
    \begin{equation}\label{eq:coltrans}
        V(D_t)y^{(l)}(t)=  \prn*{\sum_{j=1}^n V_{mj}y_j^{(q_{j}-q_{m}+l)}(t)}_m \quad (l \in [k])
    \end{equation}
    holds.
    
    To what follows, we show the following:
    \begin{enumerate}[{label={(\arabic*)}}]
        \item $(p,q)$ is also a feasible solution to $(\mathrm{D}_{f'})$, and
        \item $J_{f'}^{p,q}=UJ_f^{p,q}V$.
    \end{enumerate}
    First, we show (1).  Recall from the definition of the sigma function that at most $\sigma_f(k,m)$th derivative of $x_m$ occurs in $f_k$.  Then, from~\eqref{eq:rowtrans}, at most $(\max_k \sigma_f(k,m)+p_k-p_i)$th derivative of $x_m$ appears in $g_i$.  Next, by the change of variables defined in~\eqref{eq:coltrans}, it follows that at most $(\max_m \max_k \sigma_f(k,m)+p_k-p_i+q_j-q_m)$th derivative of $y_j$ appears in $f'_i$.  It means that $\sigma_{f'}(i,j)\leq\max_m \max_k \sigma_f(k,m)+p_k-p_i+q_j-q_m$.  Now
    \begin{align*}
        &\sigma_f(k,m)+p_k-p_i+q_j-q_m \\
        ={}&\sigma_f(k,m)-(q_m-p_k)+q_j-p_i\leq q_j-p_i
    \end{align*}
    follows from the feasibility of $p,q$ in $(\mathrm{D}_f)$.  Hence, $\sigma_{f'}(i,j)\leq q_j-p_i$ follows and $(p,q)$ is a feasible solution to $(\mathrm{D}_{f'})$ follows.

    We proceed to show (2).   from~\eqref{eq:rowtrans} and~\eqref{eq:coltrans},
  \begin{align*}
      (J_f^{p,q})_{i,j}=\pdv{{f'_i}}{y_j^{(q_j-p_i)}}&= \sum_{m=1}^n \pdv{{g}_i}{x_m^{(q_m-p_i)}}V_{mj}\\
      &= \sum_{k=1}^n \sum_{m=1}^n U_{ik}\pdv{{f}_k^{(p_k-p_i)}}{x_m^{(q_m-p_i)}}V_{mj}\\
      &= \sum_{k=1}^n \sum_{m=1}^n U_{ik}\pdv{{f}_k}{x_m^{(q_m-p_k)}}V_{mj}\\&=\sum_{k=1}^n \sum_{m=1}^n U_{ik}[J_f^{p,q}]_{km}V_{mj}
  \end{align*}
  holds.  In the fourth equality, we used Griewank's lemma~\cite[Lemma 3.7]{pryce01}.  Therefore, (2) holds.

  Now it suffices to show that $(p,q)$ is not an optimal solution to $(\mathrm{D}_{f'})$.  Conversely, assume that $(p,q)$ is an optimal solution to $(\mathrm{D}_{f'})$.  In particular, since $(\mathrm{D}_{f'})$ is feasible, there exists a maximum weight matching $\Set{(u_i,v_{\sigma(i)})}{ 1\leq i\leq n}$ for $G_{f'}$.  From the complementarity condition, $(J_{f'}^{p,q} )_{i\sigma(i)}\neq 0$ holds for all $1\leq i\leq n$.  Thus $J_{f'}^{p,q}$ has $n$ nonzero components which do not share the same row and column index.  This contradicts the fact that term-rank of $J_{f'}^{p,q}=UJ_f^{p,q}V$ is less than $n$. Therefore, $(p,q)$ is not an optimal solution to $(\mathrm{D}_{f'})$.  It means that $\hat\delta_{f'}<\hat\delta_{f}$ holds.
  \qedhere
\end{proof}
Moreover, the above transformation ensures the ``equivalence'' of DAEs in the following sense.
First, it is easily checked that $U(D_t)$ is unimodular, i.e., its inverse ${U(D_t)}^{-1}$ is a matrix whose entries are polynomials in $D_t$.
This means that the DAEs $f = 0$ and $U(D_t)f = 0$ can be expressed as linear combinations of equations or their derivatives in the other DAE, implying that their solution sets are identical.
Similarly, due to the unimodularity of $V(D_t)$, the variables $x(t)$ and $y(t)$ can also be expressed as linear combinations of components or their derivatives in the other variable, implying a one-to-one correspondence between the solutions to $f = 0$ and $f' = 0$.
We also emphasize that the equivalence holds globally since $(U, V)$ is a vanishing pair of the system Jacobian at any point.  This is in contrast to the LC-, ES-, substitution, and augmentation methods, where the certificate of singularity is valid only locally.

\begin{example}[{Continued from Example~\ref{ex:robotex-lsm}}]
    Recall that the system Jacobian~\eqref{eq:ex-jacobian} of the robotic arm DAE~\eqref{roboteq1} is approximated by the linear symbolic matrix~\eqref{eq:robot-ls}, which is already a 1CM-matrix.
    Algorithm~\ref{alg_main} finds a vanishing pair $(U,V)$ of~\eqref{eq:robot-ls}, where
    \begin{equation*}
        U=I_5, \quad
        V=\begin{pmatrix}
            1&-1&0&0&0\\
            0&1&0&0&0\\
            0&0&1&0&0\\
            0&0&0&1&1\\
            0&0&0&0&1\\
        \end{pmatrix}.
    \end{equation*}
    Given this $(U, V)$, our method transforms the DAE by changing the variable $x = (\theta_0,\theta_1,\varphi,\tau_0,\tau_1)$ into $y = (y_1, \dotsc, y_5)$ as $\theta_0 = y_1 - y_2$, $\theta_1 = y_2$, $\varphi = y_3$, $\tau_0 = y_4 + y_5$, and $\tau_1 = y_5$.
    Let $f'=0$ be the DAE after this transformation. 
    The sigma function $\sigma_{f'}$ is the same as~\eqref{eq:sigmafunc} except the following: $\sigma_{f'}(1,5)=\sigma_{f'}(2,5)=-\infty,\sigma_{f'}(4,2)=0,$ and $\sigma_{f'}(1,2)=1$.
    Consequently, one optimal solution to $(\mathrm{D}_{f'})$ can be computed as $p=(2,2,0,4,4)^\top$ and $q=(4,4,2,2,0)^\top$, then $\hat\delta_{f'}=0<2=\hat\delta_{f}$ holds as desired.
\end{example}

After regularizing a DAE $f=0$ by iterative applications of our method, we can retrieve the solution with respect to the original variable $x(t)$ as follows.
Let $f^k = 0$ $(0 \le k \le r)$ be a DAE obtained by the $k$th application of our method.
That is, $\hat{\delta}_{f^0} > \dotsb > \hat{\delta}_{f^r}$ and $f^r$ has a nonsingular system Jacobian.
We define a new DAE $f^*=0$ in $2n$ variables $(x(t), z(t))$ as $f^*_i=f^r_i(z(t),\dot{z}(t),\dotsc,t)$ and $f^*_{n+i}=x_i(t)-\prn{V^1(D_t)\cdots V^r(D_t)z(t)}_i$ for $i\in[n]$, where $V^k(D_t)$ is the matrix expressing the change of variables in the $k$th iteration.
This operation retains the nonsingularity of the system Jacobian.

\begin{theorem}\label{nonsingularmaintain}
    The system Jacobian of $f^*=0$ is nonsingular.
\end{theorem}

By solving $f^*(x(t),\dot{x}(t),\dotsc, z(t),\dot{z}(t),\dotsc,t)=0$, we can obtain a solution to the original DAE $f=f^0(x(t),\dot{x}(t),\dotsc,t)=0$. 

To prove Theorem~\ref{nonsingularmaintain}, we make some preparations.  For a DAE~\eqref{dae}, we define a polynomial matrix $A_f(s)$ as
 \begin{equation*}
A_f(s)=\prn*{\sum_{l=0}^k \pdv{f_i}{x_j^{(l)}}s^l}_{ij}.
\end{equation*}
It was shown that $\hat\delta_f$ serves as an upper bound on $\deg_s \det A_f(s)$.

\begin{lemma}[{\cite[Proposition 2.1]{murota95}}]\label{tcflem}
$\deg_s \det A_f(s) \leq \hat\delta_f$ holds.
\end{lemma}

Moreover, the nonsingularity of the system Jacobian is characterized as follows.
\begin{lemma}[{\cite[Proposition 6.2]{murota90}}]\label{tcflem2}

 Suppose $\hat\delta_f>-\infty$.  For any optimal solution $(p,q)$ to $(\mathrm{D}_f)$, $\det J_f^{p,q} \not\equiv 0$ holds if and only if $\deg_s \det A_f(s)$ is equal to $\hat\delta_f$. 
 \end{lemma}
 Note that $A_f(s)$ contains both $s$ and $x,\dot{x},\dotsc,x^{(k)},t$, and then, $\deg_s \det A_f(s)$ is defined as the maximal integer $l$ such that the coefficient of $s^l$ in $\det A_f(s)$ is not the zero function.

 Now we are ready to show Theorem~\ref{nonsingularmaintain}.
\begin{proof}[Proof of Theorem~\ref{nonsingularmaintain}]
    First, note that $A_{f^*}(s)$ can be written as 
    \begin{equation*}
        A_{f^*}(s)=\begin{pmatrix}
            O&A_{f^r}(s)\\
            I&B(s)
        \end{pmatrix}
    \end{equation*}
    where $B(s)$ is some $n\times n$ polynomial matrix.  Then, $\deg\det A_{f^*}(s)=\deg\det A_{f^{r}}(s)$ follows.  If we can show that $\hat\delta_{f^{r}}\geq \hat\delta_{f^*
    }$, then by Lemma~\ref{tcflem2}, it follows that
    \begin{equation}
        \deg\det A_{f^*}(s)=\deg\det A_{f^{r}}(s)=\hat\delta_{f^{r}}\geq \hat\delta_{f^*},
    \end{equation}
     and by Lemma~\ref{tcflem} $\deg\det A_{f^*}(s) \leq \hat\delta_{f^*}$ follows.  Hence, again by Lemma~\ref{tcflem2}, it follows that the system Jacobian of $f^*=0$ is nonsingular.
    
     Let us show that $\hat\delta_{f^{r}}\geq \hat\delta_{f^*}$.  Let $(p,q)$ be any optimal solution to $\mathrm{D}_{f^{r}}$.  Now, $p'=(p_1+z,\dotsc,p_n+z,0,\dotsc,0)^\top$ and $q'=(0,\dotsc,0,q_1+z,\dotsc,q_n+z)^\top$ is a feasible solution to $\mathrm{D}_{f^*}$, if $z$ is a sufficiently large integer.  Then, $\hat\delta_{f^*} \leq \sum_{i=1}^{2n}q'_i-\sum_{i=1}^{2n}p'_i=\sum_{i=1}^{n}q_i-\sum_{i=1}^{n}p_i=\hat\delta_{f^{r}}$ holds.  Then, it follows that if the system Jacobian of $f^{r}=0$ is nonsingular, then that of $f^*=0$ is also nonsingular.
\end{proof}

\section{Numerical Experiments}\label{sec_num}

\subsection{Experimental Setup}\label{subsec_num_setup}
In this experiment, we apply our regularization method and previous methods, the substitution and augmentation methods, to DAEs from real problems: the robotic arm, the transistor amplifier, and the ring modulator.
The robotic arm has one integer parameter $N \ge 1$ and the DAE is of size $3N+2$. The case of $N = 1$, illustrated in Example~\ref{robotex}, was derived in~\cite{deluca88}.
The other two DAEs are retrieved from~\cite{test}.
The details of these DAEs are given in Section~\ref{sec:experiment-daes}.
We confirmed that these DAEs have singular system Jacobians and the IOT-method does not work for them\footnote{More precisely, for any decomposition~\eqref{half}, there exists an optimal solution $(p,q)$ to the assignment problem for which the system Jacobian is misinterpreted as nonsingular.}.

The proposed method is implemented in MATLAB with Symbolic Math Toolbox for symbolic operations and in C++ for Algorithm~\ref{alg_main} to obtain a vanishing pair.
The arithmetic expansion to obtain linear symbolic matrices was performed by \texttt{expand} function in Symbolic Math Toolbox with setting \texttt{arithmeticOnly} option as \texttt{true}.
As explained in Remark~\ref{rem:physical-param}, we treated physical parameters as functional terms.
For \textbf{Hash} of a function $f$, a string representation obtained by \texttt{char} function is used as $H_f$.

As for the substitution and augmentation methods, we used DAEPreprocessingToolbox~\cite{git}, which is an open-source library implemented in MuPAD with a MATLAB interface. 
We confirmed the nonsingularity of the system Jacobians of the transformed DAEs by the following two ways: (i) random assignment of numerical values to the symbols in the system Jacobian and (ii) \texttt{rank} function in Symbolic Math Toolbox\footnote{\url{https://www.mathworks.com/help/symbolic/rank.html} (accessed February 1, 2024).}.
All experiments were conducted using MATLAB 2020a and GCC 8.1.0 on a laptop with Intel Corei7-8565U CPU and \SI{16}{GB} RAM.

\subsection{DAEs Used in Experiments}\label{sec:experiment-daes}
\paragraph{Robotic Arm.}
We consider a planar robotic arm equipped with $N+2$ joints, $N+1$ arms, and $N+1$ motors, illustrated in Figure~\ref{fig:robotfign}.
The arms and the motors are numbered from $0$ to $N$.
Motor 0 is inelastic and the other motors are elastic.
The joints consist of $N+1$ numbered joints from $0$ to $N$ and one special joint denoted as joint $\mathrm{E}$, which is a massless joint equipped with motors $1$ to $N$.
Joint 0 is the fulcrum fixed at the origin and is equipped with motor $0$.
For $j = 0, \dotsc, N$, arm $j$ connects joints $j$ and $\mathrm{E}$, and is driven by motor $j$.
The mass and the moment of inertia of motor $j(\geq 0)$ are denoted by $m_j$ and $J_j$, respectively.
Similarly, the mass and the moment of inertia of joint $j(\geq 1)$ are written as $m'_j$ and $J'_j$, respectively.
We assume that every arm is massless.
The length of arm $j$ is $l_j$ for $0 \leq j \leq N$.
The elastic constant and the transmission ratio of elastic motor $j$ are written as $K_j$ and $R_j$, respectively, for $j \in [N]$.
We denote the angle between the base and arm 0 as $\theta_0$, the angle between arms 0 and $j$ as $\theta_j$ for $j \in [N]$, the rotation angle of motor $j$ as $\varphi_j$ for $j \in [N]$, and the torque on motor $j$ as $\tau_j$ for $j = 0, \dotsc, N$.

We consider the problem of controlling the torques $\tau_0, \dotsc, \tau_N$ to align the horizontal coordinate of joint $\mathrm{E}$ at time $t$ as $p_0(t)$ and that of joint $j$ as $p_j(t)$ for $j \in [N]$.
The DAE expressing the dynamics of $x\coloneqq[\theta_0,\dotsc,\theta_N,\varphi_1,\dotsc,\varphi_N,\tau_0,\dotsc,\tau_N]^\top$ is given by
\begin{gather}
    \left\{\begin{aligned}\ddot{x}+\frac{1}{Q} \begin{pNiceMatrix}
        1    & a^\top & -\mathbf{1}_N^\top\\
        a    & Q\diag\{\beta_j^{-1}\}+aa^\top & a\mathbf{1}_N^\top\\
        -\mathbf{1}_N & \mathbf{1}_Na^\top & Q\diag\{J_j^{-1}\}+\mathbf{1}_N\mathbf{1}_N^\top
    \end{pNiceMatrix}z&=0,\\
    l_0\cos\theta_0=p_0(t)\smash{,}\phantom{\quad (j \in [N])} \\
    l_0\cos\theta_0+l_j\cos\theta_j=p_j(t) \quad (j \in [N]),
    \end{aligned}\right.\label{roboteq} \\
\shortintertext{where}
    Q=\alpha+2\sum_{j=1}^N\gamma_j\cos\theta_j-\sum_{j=1}^N J_j-\sum_{j=1}^N \frac{1}{\beta_j}(\beta_j+\gamma_j\cos\theta_j)^2, \notag\\
    z=\begin{pmatrix}
        -\sum_{j=1}^N (2\dot\theta_0\dot\theta_j-\dot\theta_j^2)\gamma_j\sin\theta_j-\tau_0\\
        \left[\gamma_j\dot\theta_0^2\sin\theta_j+K_j(\theta_j-\varphi/R_j)\right]_j\\
        \left[K_j/R_j \cdot (\varphi/R_j-\theta_j)-\tau_j \right]_j
    \end{pmatrix},a=\left[-1 - \gamma_j / \beta_j \cdot \cos \theta_j \right]_j, \notag\\
    \alpha=J_0+\sum_{j=1}^N(J_j+m_jl_0^2+m'_j(l_0^2+l_j^2)+J_j'), \beta_j=m_j'l_j^2+J_j',  \gamma_j = m_j'l_0l_j, \notag
\end{gather}
and $\mathbf{1}_N$ denotes the all-one vector of dimension $N$.
In the experiment, we set $p_0(t)=l_0\cos(1-\mathrm{e}^t)$ and $p_j(t)=l_0\cos(1-\mathrm{e}^t)+l_j\cos(1-jt)$ for $j \in [N]$.

\begin{figure}
    \centering
    \includegraphics[width=12cm]{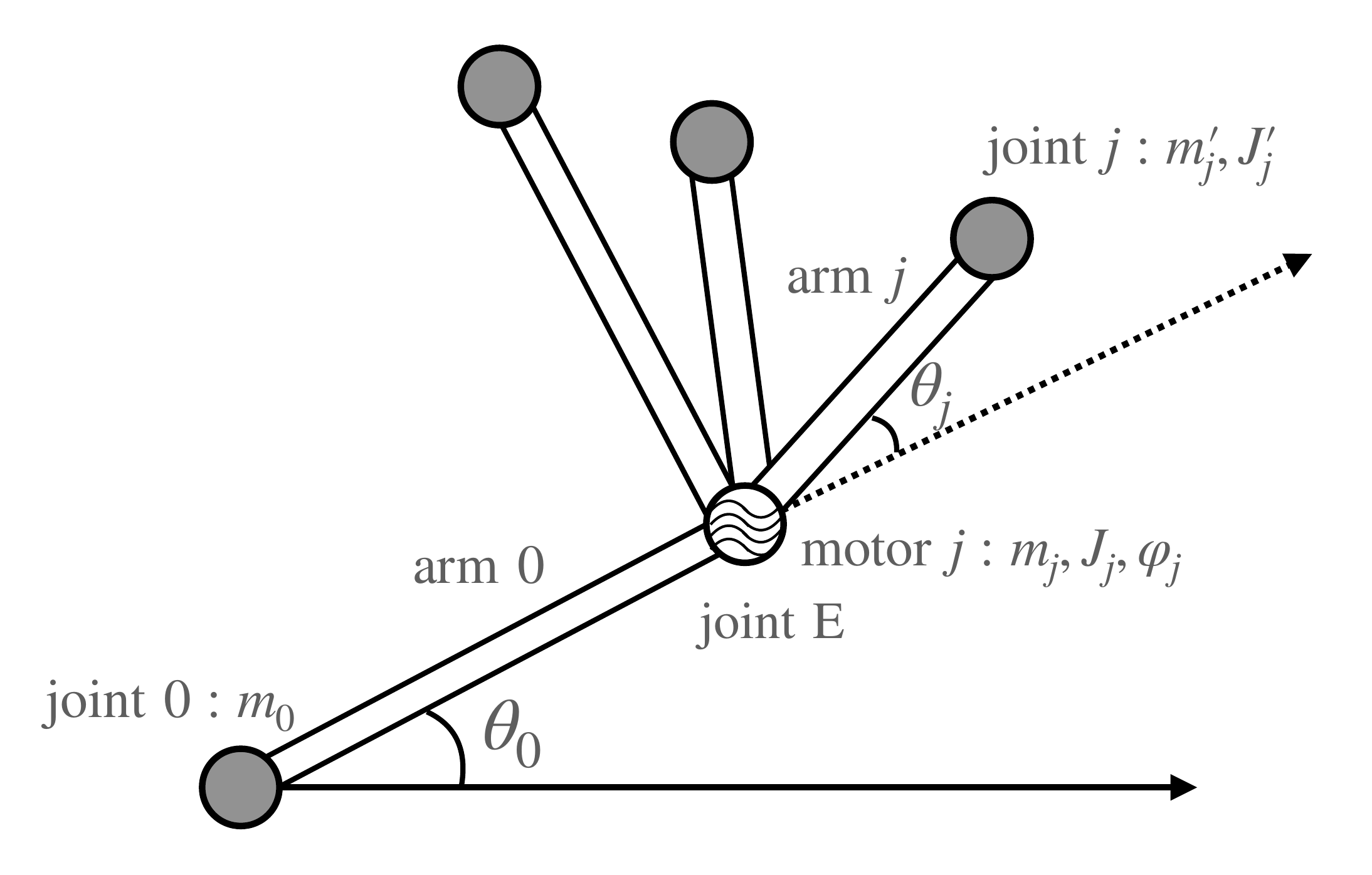}
    \caption{The robotic arm whose behavior is described by~\eqref{roboteq}.}\label{fig:robotfign}
\end{figure}

\paragraph{Transistor amplifier~(Index 1).}
A DAE in size 8
\begin{align*}
  \left\{\begin{aligned}
    C_1(\dot{x}_1-\dot{x}_2) + (x_1-U_e(t))/R_0 &= 0, \\
    - C_1(\dot{x}_1 - \dot{x}_2) - U_b/R_2 + x_2 (1/R_1 + 1/R_2) - (\alpha - 1)e(x_2 - x_3) &= 0, \\
     C_2 \dot{x}_3 + x_3/R_3 - e(x_2 - x_3) &= 0, \\
     C_3 (\dot{x}_4 - \dot{x}_5) + (x_4 - U_b)/R_4 + \alpha  e(x_2 - x_3) &= 0, \\
     - C_3 (\dot{x}_4 - \dot{x}_5) - U_b/R_6 + x_5 (1/R_5 + 1/R_6) - (\alpha  - 1)e(x_5 - x_6) &= 0, \\
     C_4 \dot{x}_6  + x_6/R_7- e(x_5 - x_6) &= 0, \\
     C_5 (\dot{x}_7 - \dot{x}_8) + (x_7 - U_b)/R_8 + \alpha  e(x_5 - x_6) &= 0, \\
     - C_5 (\dot{x}_7 - \dot{x}_8) + x_8/R_9 &= 0
  \end{aligned}\right.
\end{align*}
represents the behavior of a transistor amplifier, where $e(x)=\beta(\exp(x/U_F)-1)$ and $U_e(t)=0.1\sin(200\pi t)$.
Here, $C_1, \dotsc, C_5$, $R_0, \dotsc, R_9$, $U_b,\alpha$, and $\beta$ are physical parameters, and are treated in the experiment as independent symbols.  See~\cite{test} for the derivation.

\paragraph{Ring modulator~(Index 2).}
A nonlinear DAE in size 15
\begin{align*}
      \left\{\begin{aligned}
        \dot{x}_1 + (x_1/R - x_{8} + 0.5x_{10} - 0.5x_{11} - x_{14})/C &= 0, \\
        \dot{x}_2 + (x_{2}/R - x_{9} + 0.5x_{12} - 0.5x_{13} - x_{15})/C &= 0, \\
        x_{10} - q(U_{D1}) + q(U_{D4}) &= 0, \\
        x_{11} - q(U_{D2}) + q(U_{D3}) &= 0, \\
        x_{12} + q(U_{D1}) - q(U_{D3}) &= 0, \\
        x_{13} + q(U_{D2}) - q(U_{D4})  &= 0, \\
        \dot{x}_7 + (x_{7}/R_p - q(U_{D1}) - q(U_{D2})  + q(U_{D3}) + q(U_{D4}))/C_p &= 0, \\
        \dot{x}_8 + x_{1}/L_h &= 0, \\
        \dot{x}_9 + x_{2}/L_h &= 0, \\
        \dot{x}_{10} + (-0.5x_1 + x_3 + R_{g2} x_{10})/L_{s2} &= 0, \\
        \dot{x}_{11} + (0.5x_1 - x_4 + R_{g3} x_{11})/L_{s3} &= 0, \\
        \dot{x}_{12} + (-0.5x_2 + x_5 + R_{g2} x_{12})/L_{s2} &= 0, \\
        \dot{x}_{13} + (0.5x_2 - x_6 + R_{g3} x_{13})/L_{s3} &= 0, \\
        \dot{x}_{14} + (x_1 + (R_{g1} + R_i)x_{14} - U_{\mathrm{in1}}(t) )/L_{s1} &= 0, \\
        \dot{x}_{15} + (x_2 + (R_c + R_{g1})x_{15} )/L_{s1} &= 0
      \end{aligned}\right.
    \end{align*}
    describes the behavior of a ring modulator.  The above system is obtained by setting $C_s = 0$ in the equation in~\cite{test}.  Here, $R,R_c,R_p,R_i,R_{g1}$, $R_{g2},R_{g3},C,C_p,L_h,L_{s1},L_{s2}$, and $L_{s3}$ are physical constants and
    \begin{align*}
      &U_{D1} = x_3 - x_5 - x_7 - U_{\mathrm{in2}}(t),
      \quad
      U_{D2} = -x_4 + x_6 - x_7 - U_{\mathrm{in2}}(t),
      \\&
      U_{D3} =  x_4 + x_5 + x_7 + U_{\mathrm{in2}}(t),
      \quad
      U_{D4} = -x_3 - x_6 + x_7 + U_{\mathrm{in2}}(t),
      \\&
      q(U) = \gamma \prn{\e^{\delta U} - 1},
      \,
      U_{\mathrm{in1}}(t) = 0.5 \sin 2000 \pi t,
      \,
      U_{\mathrm{in2}}(t) = 2 \sin 20000 \pi t.
    \end{align*}
\subsection{Experimental Results}\label{subsec_num_res}
\begin{table}[t]
 \centering
 \caption{Comparison of the running time [sec] required to regularize the system Jacobian. 
 Annotation (*) indicates that the method did not stop within an hour.}\label{tbl:result}
  \begin{tabular}{crrrrr}\toprule
   DAE &
   \multicolumn{1}{c}{$N$} &
   \multicolumn{1}{c}{Sub.} &
   \multicolumn{1}{c}{Aug.} &
   \multicolumn{1}{c}{\textbf{Proposed}} \\\midrule
   Transistor amplifier & --- & 0.5 &0.3&3.0 \\
   Ring modulator & --- &(*)&0.9&4.6 \\
   Robotic arm & 1 &4.0&0.6&1.6 \\
   & 2 &5.3&1.0&3.1 \\
   & $3$ &8.1&1.8&5.4 \\
   & $4$ &14.3&3.3&8.8 \\
   & $5$ &25.5&6.7&12.7 \\
   & $6$ &40.6&13.6&18.9 \\
   & $7$ &68.3&23.1&26.4 \\
   & $8$ &100.3&39.5&37.7 \\
   & $9$ &149.8&62.3&53.2 \\
   & $10$ &236.6&103.7&44.6 \\
   & $11$ &298.7&177.6&59.0 \\
   & $12$ &401.1&292.2&76.6 \\
   & $13$ &598.6&635.2&100.6 \\
   & $14$ &816.5&1681.0&200.9 \\
   & $15$ &1254.2&3273.1&262.0 \\
   & $16$ &1632.2&(*)&334.7 \\
   & $17$ &2358.7&(*)& 431.2\\
   & $18$ &2713.6&(*)& 559.2\\
   & $19$ &(*)&(*)& 728.2\\
   & $20$ &(*)&(*)&936.7 \\
   \bottomrule
  \end{tabular}
\end{table}
\begin{figure}[t]
\centering
\includegraphics[width=12cm]{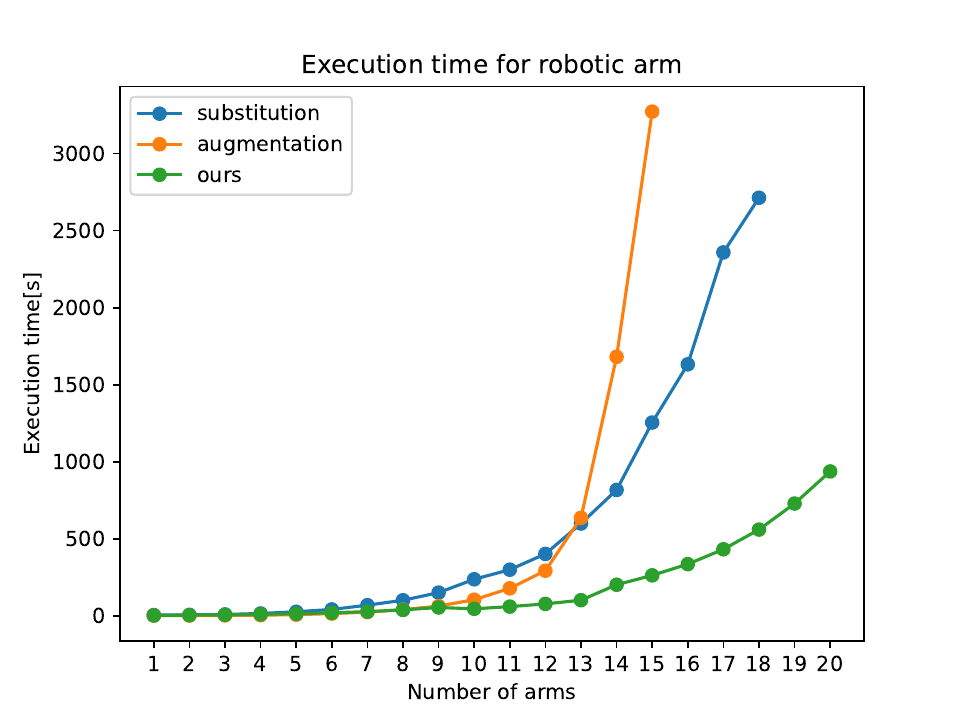}
\caption{Running time for the robotic arm DAE~\eqref{roboteq}}\label{fig:exp_robot}
\end{figure}
Our proposed method successfully returned DAEs with nonsingular system Jacobians for all the DAEs used in the experiment.
Table~\ref{tbl:result} shows the computation time for each method.  For the robotic arm, we conducted the experiment by changing the number $N$ of arms from $1$ to $20$.
Figure~\ref{fig:exp_robot} plots the graph of the running times versus $N$ in the DAE~\eqref{roboteq}.
For the robotic arm DAE with large $N$, our method runs much faster than existing methods.
We also observed the low-rank structures of the coefficient matrices of the linear symbolic matrix $J_\mathrm{LS}$ approximating the system Jacobians.  For the robotic arm and the transistor amplifier, the ranks of the coefficient matrices of $J_\mathrm{LS}$ are all one.  Hence, we did not need to apply the rank-one decomposition~\eqref{rb} to obtain a 1CM-matrix $J_\mathrm{1CM}$.  For the ring modulator, eight out of ten coefficient matrices have rank one but the other two are of rank two.
Fortunately, these rank-two matrices do not contribute to the singularity of $J_\mathrm{LS}$, hence $\rank J_\mathrm{1CM} = \rank J_\mathrm{LS}$ hold.

\section{Conclusion}\label{sec_conclusion}
This paper has proposed a fast and globally-equivalent regularization method for nonlinear DAEs, extending the IOT method.
In Phase~3 of the combinatorial relaxation framework, our method approximates the system Jacobian by a 1CM-matrix, for which we have developed a fast combinatorial algorithm to obtain a vanishing pair.
Though numerical experiments, we have confirmed that our method successfully regularizes DAEs from real instances and runs fast for large-scale DAEs.

A future challenge lies in proposing more efficient and provable method for converting the system Jacobian into a linear symbolic matrix. Another future work involves developing a combinatorial algorithm to obtain a certificate of singularity for linear symbolic matrices with low-rank coefficients, without lying on the rank-one decomposition.

\section*{Acknowledgments}
The authors thank Satoru Iwata, Mizuyo Takamatsu, and Tasuku Soma for their helpful comments and discussions.
This work is supported by JSPS KAKENHI Grant Number JP22K17853 and JST ERATO Grant Number JPMJER1903.

\clearpage
\appendix
\section*{Appendix}

\section{DAEs with System Jacobians Approximated by 1CM-Matrices}\label{app-1cm}

In this section, we consider a class of DAEs which can be expressed as a composition of linear and nonlinear expressions in certain form.
DAEs in this class have a property that the linear symbolic matrices $J_\mathrm{LS}$ constructed from the system Jacobians can be 1CM-matrices.
We also explain that this class subsumes DAEs arising from several real problems: electrical circuits modeled via the modified nodal analysis, interacting multi-body systems, and their linear transformations.

\subsection{DAE Form Description}

Consider a DAE $f = 0$ in size $n$ for an unknown variable $x = (x_1,\dotsc,x_n)^\top$ in the form of
\begin{align}\label{rank1}
    \begin{aligned}
    &f(x,\dotsc,x^{(k)},t)\\={}&g(t)+\sum_{l=0}^k B_l x^{(l)}+ \sum_{l=0}^k\sum_{m=1}^{N(l)} h_{lm} \prn[\big]{{a_{lm}}^\top x^{(l)}} b_{lm},
    \end{aligned}
\end{align}
where $g:\mathbb{R}\to\mathbb{R}^n$ is a smooth function, $\{B_l\}$ are $n\times n$ constant matrices, $N(l)$ is a non-negative integer, $\{a_{lm}\}$ and $\{b_{lm}\}$ are $n$-dimensional constant vectors, and $h_{lm}:\mathbb{R}\to \mathbb{R}$ is a smooth function for $l = 0, \dotsc, k$ and $m \in [N(l)]$.
We can show that the system Jacobian of the DAE~\eqref{rank1} has the rank-one structure as follows.

\begin{theorem}\label{1cmgeneral}
    The system Jacobian of~\eqref{rank1} can be expanded as
    \begin{equation}\label{1cmexpand}
J_f^{p,q}=A_0+\sum_{l=0}^k\sum_{m=1}^{N(l)}h'_{lm} \prn[\big]{{a_{lm}}^\top x^{(l)}} A^h_{lm},
    \end{equation}
    where $h'_{lm}$ is the derivative of $h_{lm}$, $A_0$ is a constant matrix of arbitrary rank, and $\{A^h_{lm}\}$ are constant matrices with ranks at most one.
\end{theorem}

By Theorem~\ref{1cmgeneral}, the linear symbolic matrix \begin{equation*}
J_{\mathrm{LS}}=A_0+\sum_{l=0}^k\sum_{m=1}^{N(l)}\alpha_{lm}A^h_{lm}
\end{equation*} obtained by replacing each $h'_{lm}\prn[\big]{{a_{lm}}^\top x^{(l)}}$ with a distinct symbol $\alpha_{lm}$ is a 1CM-matrix, i.e., $J_\mathrm{LS} = J_\mathrm{1CM}$ holds.
If the nonzero nonlinear terms are generic, that is, $\Set{h'_{lm}\prn[\big]{{a_{lm}}^\top x^{(l)}} }{h'_{lm}\prn[\big]{{a_{lm}}^\top x^{(l)}} \not\equiv 0}$ do not have any algebraic relationship, then the nonsingularity of $J_f^{p,q}$ and $J_{\mathrm{1CM}}$ are equivalent.

We give the proof of Theorem~\ref{1cmgeneral} in Section~\ref{sec:1cmgeneral-proof} after presenting examples of DAEs in the form of~\eqref{rank1} in Section~\ref{sec:1cm-examples}.

\subsection{Examples}\label{sec:1cm-examples}
\paragraph{Modified nodal analysis for electrical circuits.}
Modified nodal analysis~(MNA,~\cite{ho75}) is a widely-used modeling approach for electrical circuits.
Consider an electrical circuit consisting of nonlinear resistors, capacitors, inductors, independent current sources, and voltage sources.
If the potential difference between the ends of the $k$th resistor (capacitor) is $E$, then the current passing through the element is written as $r_k(E)$ (resp.\ $c_k(\dot{E})$), where $r_k$ (resp.\ $c_k$) is some nonlinear function describing the characteristics of the element.
Similarly, if the current passing through the $k$th inductor is $I$, then the potential difference between the ends of the inductor is written as $l_k(\dot{I})$, where $l_k$ denotes some nonlinear characteristics.  Suppose that the circuit has $N+1$ nodes, comprised of one reference node and other $N$ nodes.  The sum of the numbers of inductors and voltage sources (we call them as voltage-controlled elements) is denoted by $M$.  We denote the potential of the $k$th node as $E_k$ for $k \in [N]$ and set the potential of the reference node as 0.  We also write the current passing through the $k$th voltage-controlled element as $I_k$ for $k \in [M]$.  The modified nodal analysis constructs a DAE in size $N+M$ composed of
\begin{itemize}
    \item $N$ equations which formulate Kirchhoff's circuit law (i.e.\ the sum of the current flowing into a node is zero) at non-reference nodes, and
    \item $M$ equations expressing the potential difference of the ends of the $k$th voltage-controlled element using $I_k$.
\end{itemize}

\begin{figure}
    \centering
    \includegraphics[width=4cm]{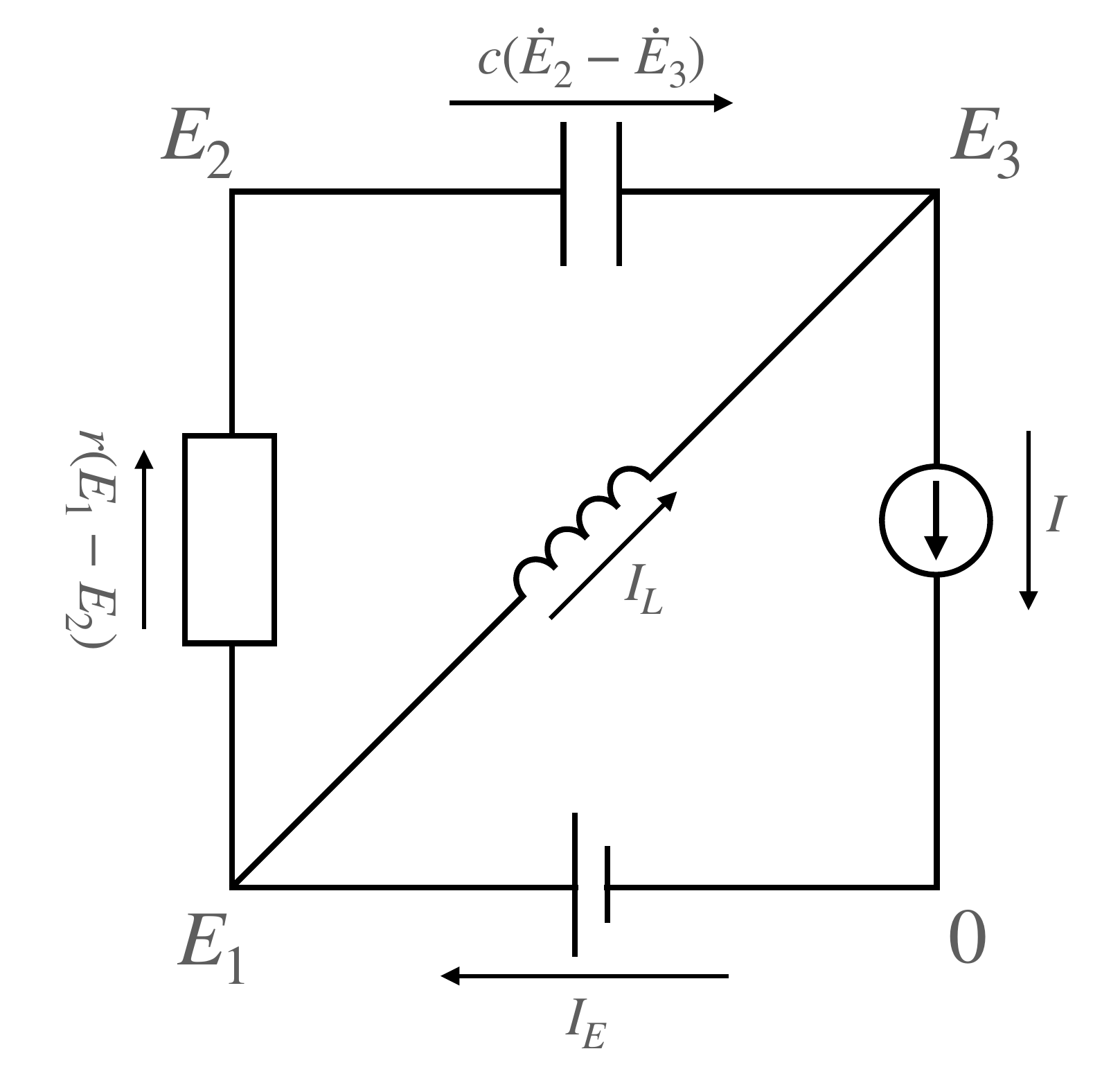}
    \caption{Electrical circuit described by~\eqref{mna}.}\label{elec}
\end{figure}

We can check that the resulting DAE can be expressed in the form of~\eqref{rank1}, because the DAE contains only nonlinear terms expressed as $r_k(E_{s_k}-E_{t_k}),c_k(\dot{E}_{s_k}-\dot{E}_{t_k}),$ or $l_k(\dot{I}_{k})$, where $s_k$ and $t_k$ denotes the ends of the $k$th voltage-controlled element and ${I}_{k}$ denotes the current passing through the $k$th inductor.

\begin{example}\label{ex:mna}
    Consider a circuit with $N=3$ and $M=2$ illustrated in Figure~\ref{elec}.
    For this circuit, the MNA outputs the following DAE in variable $x = (E_1,E_2,E_3,I_E,I_L)$:
    \begin{align}\label{mna}
        \left\{\begin{aligned}
            I_E-I_L-r(E_1-E_2)&=0,\\
            r(E_1-E_2)-c(\dot{E}_2-\dot{E}_3)&=0,\\
            c(\dot{E}_2-\dot{E}_3)+I_L-I&=0,\\
            E_1-E&=0,\\
            l(\dot{I}_L)-E_3+E_1&=0,
        \end{aligned}\right.
    \end{align}
    where $I$ and $E$ are the current of the current source and the electromotive force of the voltage source, respectively.
    The DAE~\eqref{mna} can be written in the form~\eqref{rank1}, where $k = 1$, $N(0) = 1$, $N(1) = 2$,
    \begin{gather*}
        g(t)=\begin{pmatrix}
            0\\0\\-I\\-E\\0
        \end{pmatrix},\quad B_0=\begin{pmatrix}
            0&0&0&1&-1\\
            0&0&0&0&0\\
            0&0&0&0&1\\
            1&0&0&0&0\\
            1&0&-1&0&0
        \end{pmatrix}, \quad B_1 = O,\\
        h_{01}(z)=r(z), \quad a_{01}=(1,-1,0,0,0)^\top, \quad b_{01}=(-1,1,0,0,0)^\top,\\
        h_{11}(z)=c(z), \quad a_{11}=(0,1,-1,0,0)^\top, \quad b_{11}=(0,-1,1,0,0)^\top,\\
        h_{12}(z)=l(z), \quad a_{12}=(0,0,0,0,1)^\top, \quad b_{12}=(0,0,0,0,1)^\top.
    \end{gather*}

    We check that Theorem~\ref{1cmgeneral} actually holds for the DAE~\eqref{mna}.
    An optimal solution to the dual problem of the assignment problem is $p=(0,0,0,0,0)^\top$ and $q=(0,1,1,0,1)^\top$, and the system Jacobian is calculated as
    \begin{equation*}
        J=\begin{pmatrix}
            -r'(E_1-E_2)&0&0&1&0\\
            r'(E_1-E_2)&-c'(\dot{E}_2-\dot{E}_3)&c'(\dot{E}_2-\dot{E}_3)&0&0\\
            0&c'(\dot{E}_2-\dot{E}_3)&-c'(\dot{E}_2-\dot{E}_3)&0&0\\
            1&0&0&0&0\\
            1&0&0&0&l'(\dot{I}_L)
        \end{pmatrix},
    \end{equation*}
    which is singular.
    By mapping $r'(E_1-E_2)\mapsto \alpha_1,c'(\dot{E}_2-\dot{E}_3)\mapsto\alpha_2,$ and $l'(\dot{I}_L)\mapsto \alpha_3$ in $J$, we can construct a linear symbolic matrix
    \begin{equation*}
        J_{\mathrm{LS}}=\begin{pmatrix}
            -\alpha_1&0&0&1&0\\
            \alpha_1&-\alpha_2&\alpha_2&0&0\\
            0&\alpha_2&-\alpha_2&0&0\\
            1&0&0&0&0\\
            1&0&0&0&\alpha_3
        \end{pmatrix},
    \end{equation*}
    which is a 1CM-matrix.
    We can also confirm that the matrix $J_\mathrm{1CM}$ is still singular.
\end{example}

\paragraph{Multi-body dynamical systems.}
Let us consider a multi-body system in which $n$ mass points $X_1,\dotsc,X_n$ interacts.  Now assume that these mass points have unit mass and are arranged in one-dimensional space: extension to higher dimensions is easy.  Let $x_1,\dotsc,x_n$ be the coordinates of $X_1,\dotsc,X_n$, respectively.
Let $G=(X,E)$ be an undirected graph whose vertex set is $X=\{X_1,\dotsc,X_n\}$.
We add an edge $\{X_i, X_j\}$ to $E$ if $X_i$ and $X_j$ interact each other, and assume that the force between $X_i$ and $X_j$ can be written as $k_{ij}(x_i-x_j)$ or $k_{ij}(\dot x_i-\dot x_j)$ using some nonlinear smooth function $k_{ij}$.  Many types of interactions, such as elastic force, electrostatic force, and viscous force by dampers, can be written in this form.
Without loss of generality, we assume that the force $k_{ij}(x_i-x_j)$ or $k_{ij}(\dot x_i-\dot x_j)$ is applied to $X_i$ if $i<j$.
Then, by Newton's Third Law, the force of the opposite direction, i.e., $-k_{ij}(x_i-x_j)$ or $-k_{ij}(\dot x_i-\dot x_j)$, is applied to $X_j$.  
The equation of motion is written as follows:
\begin{align}\label{eq:multibody-ode}
    \left\{ 
    \begin{aligned}
        \ddot{x_i}&=\sum_{i<j,(X_i,X_j)\in E}[k_{ij}(x_i-x_j)~\text{or}~k_{ij}(\dot x_i-\dot x_j)]\\&-\sum_{j<i,(X_i,X_j)\in E}[k_{ij}(x_i-x_j)~\text{or}~k_{ij}(\dot x_i-\dot x_j)]~(i\in[n]).
    \end{aligned}
    \right.
\end{align}
We can check that the DAE~\eqref{eq:multibody-ode} is in the form of~\eqref{rank1}, though it is indeed an ODE.
DAEs appear, for example, in the following situation: consider giving $d (\leq n)$-dimensional input $u(t)$ to the system and control some of $\{x_1,\dotsc,x_n\}$.  Let $p(t)$ be a $d$-dimensional reference signal.  Our goal is to control $Bx$ to be $p(t)$, giving external forces $Au(t)$ to the system, where $x=[x_1,\dotsc,x_n]^\top$, $A\in\mathbb{R}^{n\times d}$, and $B\in\mathbb{R}^{d\times n}$.  We obtain a DAE in size $n+d$ for $x(t)$ and $u(t)$ as follows:
\begin{align}\label{control}
    \left\{ 
    \begin{aligned}
        \ddot{x_i}&=\sum_{i<j,(X_i,X_j)\in E}[k_{ij}(x_i-x_j)~\text{or}~k_{ij}(\dot x_i-\dot x_j)]\\&-\sum_{j<i,(X_i,X_j)\in E}[k_{ij}(x_i-x_j)~\text{or}~k_{ij}(\dot x_i-\dot x_j)]+(Au)_i~(i\in[n]),\\
        Bx&=p.
    \end{aligned}
    \right.
\end{align}
The DAE~\eqref{control} can also be written in the form of~\eqref{rank1}.

\paragraph{Robustness to linear transformations}
The DAE~\eqref{rank1} has another desirable property that it is robust to linear transformations of equations and variables.
More concretely, for a DAE $f=0$ written as~\eqref{rank1}, consider applying the following transformations to get a new DAE $f'(y, \dot{y}, \dotsc, t) = 0$: $f'_i=\sum_{j=1}^n c_{ij}f_j$ and $y_i=\sum_{j=1}^n d_{ij}x_j$ for $i \in [N]$ with some constants $\{c_{ij}\}$ and $\{d_{ij}\}$; we assume that matrices $C=\prn*{c_{ij}}_{ij}$ and $D=\prn*{d_{ij}}_{ij}$ are invertible.
Then, we can check that $f'=0$ is written in the form of~\eqref{rank1} as
\begin{align*}
    &f'(y,\dot{y},\dotsc,t)\\
    ={} &Cg(t)+\sum_{l=0}^k (CB_l D^{-1})y^{(l)}+\sum_{l=0}^k\sum_{m=1}^{N(l)}h_{lm} \prn[\big]{({a_{lm}}^\top D^{-1}) y^{(l)}} (Cb_{lm}).
\end{align*}
By Theorem~\ref{1cmgeneral}, the system Jacobian of the transformed DAE $f' = 0$ can also be approximated by a 1CM-matrix.
This kind of robustness does not hold in the approximation by layered mixed matrices, since they are not closed under linear transformations.

\subsection{Proof of Theorem~\ref{1cmgeneral}}\label{sec:1cmgeneral-proof}
\begin{proof}[Proof of Theorem~\ref{1cmgeneral}]
    It can be directly checked that system Jacobian $J_f^{p,q}$ of~\eqref{rank1} is expanded as ~\eqref{1cmexpand} using some constant matrices $A_0$ and $\{A^h_{lm}\}$. 
    We show that the ranks of $\{A^h_{lm}\}$ are at most one. 
    For a matrix $A$ in size $n$, its \emph{support} is defined as $S\coloneqq\{(i,j)\mid A_{ij}\neq 0\}\subseteq[n]^2$.  We call a support $S$ is in the \emph{block form} if there exist $I\subseteq[n]$ and $J\subseteq[n]$ such that $S=I\times J$.  Now, the support of $A^h_{lm}$ is a subset of the support of $C_{lm}\coloneqq b_{lm}a_{lm}^\top$, because from the definition of the system Jacobian, the $(i,j)$ entry of $A^h_{lm}$ is equal to that of $C_{lm}$ if $q_j-p_i=l$ holds and otherwise equal to zero.
 In the following, we show that the support of $A^h_{lm}$ is in the block form, which implies $\rank A^h_{lm} \le 1$.
 
 Let $S$ be the support of $C_{lm}$.
 Clearly, $S$ is in the block form because $\rank C_{lm}\leq 1$, and without loss of generality we can assume that $S=[a]\times[b]$ for some $0\leq a,b\leq n$, where $[0]:=\emptyset$.  We can also assume that $p_1\geq \cdots \geq p_a$ and $q_1\leq \cdots \leq q_b$ hold.  Note that if $(i,j)\in S$, then $\sigma(i,j)\geq l$ holds.  By the definition of system Jacobian, $(i,j)$ entry of $A^h_{lm}$ is nonzero if and only if $q_j-p_i=l$ holds.
    If $q_1-p_1>l$, then for all $i,j\in S$, $q_j-p_i\geq q_1-p_1>l$ holds.  Then, the support of $A^h_{lm}$ is empty and is in the block form.  Let us assume that $q_1-p_1=l$.  For $(i,j)\in S$, if $q_j>q_1$ or $p_i<p_1$ holds, then $q_j-p_i>q_1-p_1=l$ and $(i,j)$
 is not an element of the support of $A^h_{lm}$.  This means that the support of $A^h_{lm}$ can be written as $\{(i,j)\in S\mid q_j=q_1,p_i=p_1\}$, and is in the block form. \qedhere
\end{proof}

\end{document}